\documentclass[a4paper]{amsart}\usepackage[a4paper,margin=2.5cm,includefoot]{geometry}
\usepackage[utf8]{inputenc}
\usepackage[english]{babel}
\usepackage[T1]{fontenc}

\usepackage{amsmath,amssymb,mathtools,bm}
\usepackage{csquotes}
\usepackage{amstext}
\usepackage{amsthm}
\usepackage{bbm}
\usepackage{enumerate}
\usepackage{hyperref}
\usepackage[nameinlink,capitalise]{cleveref}
\usepackage{braket}
\usepackage{dsfont}
\usepackage{color}
\usepackage{tikz}
\usepackage{pgfplots}

\makeatletter
\def\@seccntformat#1{%
	\protect\textup{\protect\@secnumfont
		\ifnum\pdfstrcmp{subsection}{#1}=0 \bfseries\fi
		\ifnum\pdfstrcmp{subsubsection}{#1}=0 \itshape\fi
		\csname the#1\endcsname
		\protect\@secnumpunct
	}%
}
\renewcommand{\@upn}{}
\DeclareRobustCommand{\crefnosort}[1]{%
	\begingroup\@cref@sortfalse\cref{#1}\endgroup
}
\makeatother

\numberwithin{equation}{section}
\newtheorem{thm}{Theorem}[section]
\newtheorem{lem}[thm]{Lemma}
\AddToHook{env/lem/begin}{\crefalias{thm}{lem}}
\newtheorem{prop}[thm]{Proposition}
\AddToHook{env/prop/begin}{\crefalias{thm}{prop}}
\newtheorem{cor}[thm]{Corollary}
\AddToHook{env/cor/begin}{\crefalias{thm}{cor}}

\theoremstyle{definition}

\AddToHook{env/defn/begin}{\crefalias{thm}{defn}}
\newtheorem{hyp}{Hypothesis}
\AddToHook{env/hyp/begin}{\crefalias{thm}{hyp}}
\renewcommand*{\thehyp}{\Alph{hyp}}

\theoremstyle{remark}
\newtheorem{rem}{Remark}[section]
\newtheorem{ex}[rem]{Example}
\AddToHook{env/ex/begin}{\crefalias{rem}{ex}}

\crefname{hyp}{Hypothesis}{Hypotheses}
\Crefname{hyp}{Hypothesis}{Hypotheses}
\crefname{lem}{Lemma}{Lemmas}
\Crefname{lem}{Lemma}{Lemmas}
\crefname{thm}{Theorem}{Theorems}
\Crefname{thm}{Theorem}{Theorems}
\crefname{prop}{Proposition}{Propositions}
\Crefname{prop}{Proposition}{Propositions}
\crefname{enumi}{}{}
\Crefname{enumi}{}{}
\creflabelformat{enumi}{#2(#1)#3}
\crefname{equation}{}{}
\Crefname{equation}{}{}
\crefname{rem}{Remark}{Remarks}
\Crefname{rem}{Remark}{Remarks}

\makeatletter
\renewcommand{\@upn}{} 
\makeatother

\usepackage[inline]{enumitem}
\newlist{enumthm}{enumerate}{1} 
\setlist[enumthm]{label=\upshape(\roman*),ref=\thethm~(\roman*)}  
\crefalias{enumthmi}{thm} 
\newlist{enumcor}{enumerate}{1}
\setlist[enumcor]{label=\upshape(\roman*),ref=\thecor~(\roman*)}
\crefalias{enumcori}{cor}
\newlist{enumlem}{enumerate}{1}
\setlist[enumlem]{label=\upshape(\roman*),ref=\thelem~(\roman*)}
\crefalias{enumlemi}{lem}
\newlist{enumprop}{enumerate}{1}
\setlist[enumprop]{label=\upshape(\roman*),ref=\theprop~(\roman*)}
\crefalias{enumpropi}{prop}
\newlist{enumhyp}{enumerate}{1}
\setlist[enumhyp]{label=\upshape(\roman*),ref=\thehyp~(\roman*)}
\crefalias{enumhypi}{hyp}
\newlist{enumproof}{enumerate*}{1}
\setlist[enumproof]{label=\upshape(\roman*)}
\newlist{enumdef}{enumerate}{1}
\setlist[enumdef]{label=\upshape(\roman*),ref=\thedefn~(\roman*)}
\crefalias{enumdefi}{defn}

\makeatletter
\newcounter{subcreftmpcnt} %
\newcommand\romansubformat[1]{(\roman{#1})} 
\def\subcref{\@ifstar\@@subcref\@subcref}
\newcommand\@subcref[2][\romansubformat]{%
	\ifcsname r@#2@cref\endcsname
	\cref@getcounter {#2}{\mylabel}%
	\setcounter{subcreftmpcnt}{\mylabel}%
	\hyperref[#2]{\romansubformat{subcreftmpcnt}}%
	\else ?? \fi}   
\newcommand\@@subcref[2][\romansubformat]{%
	\ifcsname r@#2@cref\endcsname
	\cref@getcounter {#2}{\mylabel}%
	\setcounter{subcreftmpcnt}{\mylabel}%
	\romansubformat{subcreftmpcnt}%
	\else ?? \fi}   
\makeatother

\makeatletter
\DeclareRobustCommand{\crefnosort}[1]{%
	\begingroup\@cref@sortfalse\cref{#1}\endgroup
}
\makeatother

\def\endstepsymbol{$\lozenge$}
\def\endclaimsymbol{$\lozenge$}
\newcounter{proofstep}
\AtBeginEnvironment{proof}{\setcounter{proofstep}{0}}

\crefname{proofstep}{Step}{Steps}
\Crefname{proofstep}{Step}{Steps}
\newcounter{proofclaim}
\AtBeginEnvironment{proof}{\setcounter{proofclaim}{0}}

\crefname{proofclaim}{Claim}{Claims}
\Crefname{proofclaim}{Claim}{Claims}

\newcommand{\cC}{{\mathcal C}}
\newcommand{\cD}{{\mathcal D}}\newcommand{\cF}{{\mathcal F}}

\newcommand{\cP}{{\mathcal P}}

\newcommand{\BC}{{\mathbb C}}

\newcommand{\BN}{{\mathbb N}}
\newcommand{\BR}{{\mathbb R}}

\newcommand{\BZ}{{\mathbb Z}}

\usepackage{dsfont} 

\newcommand{\DSE}{{\mathds E}}

\newcommand{\dsone}{{\mathds 1}}

\usepackage{mathrsfs}

\newcommand{\sD}{{\mathscr D}}

\newcommand{\sfc}{{\mathsf c}}
\newcommand{\sfd}{{\mathsf d}}
\newcommand{\sfi}{{\mathsf i}}

\newcommand{\sfm}{{\mathsf m}}

\newcommand{\sfs}{{\mathsf s}}

\newcommand{\sfy}{{\mathsf y}}

\newcommand{\rmc}{{\mathrm c}}
\newcommand{\rme}{{\mathrm e}}

\newcommand{\bfN}{{\mathbf N}}
\newcommand{\bfP}{{\mathbf P}}

\newcommand{\IN}{\BN}\newcommand{\IZ}{\BZ}\newcommand{\IR}{\BR}\newcommand{\IC}{\BC}
\newcommand{\N}{\BN}\newcommand{\Z}{\BZ}\newcommand{\R}{\BR}

\newcommand{\EE}{\DSE}

\newcommand{\eps}{\varepsilon}\newcommand{\ph}{\varphi}

\newcommand{\e}{\rme}\renewcommand{\i}{\sfi}\newcommand{\Id}{\dsone} \renewcommand{\d}{\sfd}

\renewcommand{\Re}{\operatorname{Re}}

\DeclareMathOperator*{\essinf}{ess\,inf}

\DeclareFontFamily{U}{mathx}{\hyphenchar\font45}
\DeclareFontShape{U}{mathx}{m}{n}{
	<5> <6> <7> <8> <9> <10>
	<10.95> <12> <14.4> <17.28> <20.74> <24.88>
	mathx10
}{}
\DeclareSymbolFont{mathx}{U}{mathx}{m}{n}
\DeclareFontSubstitution{U}{mathx}{m}{n}
\DeclareMathAccent{\widecheck}{0}{mathx}{"71}
\DeclareMathAccent{\wideparen}{0}{mathx}{"75}

\DeclareFontFamily{OMX}{MnSymbolE}{}
\DeclareFontShape{OMX}{MnSymbolE}{m}{n}{
	<-6>  MnSymbolE5
	<6-7>  MnSymbolE6
	<7-8>  MnSymbolE7
	<8-9>  MnSymbolE8
	<9-10> MnSymbolE9
	<10-12> MnSymbolE10
	<12->   MnSymbolE12}{}
\DeclareSymbolFont{mnlargesymbols}{OMX}{MnSymbolE}{m}{n}
\SetSymbolFont{mnlargesymbols}{bold}{OMX}{MnSymbolE}{b}{n}
\DeclareMathDelimiter{\llangle}{\mathopen}{mnlargesymbols}{'164}{mnlargesymbols}{'164}
\DeclareMathDelimiter{\rrangle}{\mathclose}{mnlargesymbols}{'171}{mnlargesymbols}{'171}
\DeclareMathDelimiter{\lsem}{\mathopen}{mnlargesymbols}{'102}{mnlargesymbols}{'102}
\DeclareMathDelimiter{\rsem}{\mathclose}{mnlargesymbols}{'107}{mnlargesymbols}{'107}
\DeclareMathDelimiter{\langlebar}{\mathopen}{mnlargesymbols}{'152}{mnlargesymbols}{'152}
\DeclareMathDelimiter{\ranglebar}{\mathclose}{mnlargesymbols}{'157}{mnlargesymbols}{'157}
\DeclareMathDelimiter{\lWavy}{\mathopen}{mnlargesymbols}{'137}{mnlargesymbols}{'137}
\DeclareMathDelimiter{\rWavy}{\mathopen}{mnlargesymbols}{'137}{mnlargesymbols}{'137}

\newcommand{\chr}{\mathbf 1}
\newcommand{\abs}[1]{\lvert#1\rvert}
\newcommand{\norm}[1]{\lVert#1\rVert}

\newcommand{\FGamma}{\Gamma}
\newcommand{\FS}{\cF}\newcommand{\dG}{\sfd\FGamma}

\title[Ising Phase Transition in the Spin Boson Model]{On the Ising Phase Transition\\ in the Infrared-Divergent Spin Boson Model}
\author{Volker Betz}
\address{V. Betz,  Technische Universität Darmstadt, Fachbereich Mathematik, Schloßgartenstraße 7, 64289 Darmstadt, Germany}
\email{betz@mathematik.tu-darmstadt.de}
\author{Benjamin Hinrichs}
\address{B. Hinrichs, Universit\"at Paderborn, Institut f\"ur Mathematik, Institut f\"ur Photonische Quantensysteme, Warburger Str. 100, 33098 Paderborn, Germany
	[Present Affiliation: Technische Universit\"at Braunschweig, Institut f\"ur Analysis und Algebra, Universit\"atsplatz 2, 38106 Braunschweig, Germany]}
\email{benjamin.hinrichs@tu-bs.de}
\author{Mino Nicola Kraft}\address{M.\,N. Kraft, Technische Universität Darmstadt, Fachbereich Mathematik, Schloßgartenstraße 7, 64289 Darmstadt, Germany}
\email{mino.nicola.kraft@gmail.com}
\author{Steffen Polzer}
\address{S. Polzer, Université de Genève, Section de mathématiques, Rue du Conseil-Général 7-9, 1205 Genève, Switzerland [Present Affiliation: Institute of Science and Technology Austria, Am Campus 1, 3400 Klosterneuburg, Austria]}
\email{steffen.polzer@ist.ac.at}

\usepackage{graphicx}

\usepackage{xcolor}\usepackage{cancel}

\newcommand{\1}{\mathds{1}}
\newcommand{\ee}{\e}

\begin{document}

\begin{abstract} 
	\noindent
	We prove absence of ground states in the infrared-divergent spin boson model at large coupling.
	Our key argument reduces the proof to verifying long range order 
	in the dual one-dimensional continuum Ising model, i.e., to showing that the respective two point function is lower bounded by a strictly positive constant.
	We can then use known results from percolation theory to establish long range order at large coupling.
	Combined with the known existence of ground states at small coupling,
	our result proves that the spin boson model undergoes a phase transition with respect to the coupling strength.
	We also present an expansion for the vacuum overlap of the spin boson ground state in terms of the Ising $n$-point functions,
	which implies that the phase transition is unique,
	i.e., that there is a critical coupling constant below which a ground state exists and above which none can exist.
\end{abstract}

\maketitle

\section{Introduction}

In models of light-matter interactions, the infrared catastrophe is an inherent phenomenon \cite{BlochNordsiek.1937}.
It reflects the fact that an infinite number of low-energy photons can jointly have finite energy, if the photons are massless.
Thus, already small energy fluctuations can lead to the creation of such an infinite cloud of bosons.

Infrared properties have especially been studied in the context of non-relativistic quantum field theory,
where the infrared catastrophe induces the absence of a lowest energy stable state, i.e., a ground state.
For example, in the Nelson model introduced in \cite{Nelson.1964}, the infrared problem was already pointed out in \cite{Frohlich.1973}. 
Proofs for absence of ground states were given in \cite{LorincziMinlosSpohn.2002,DerezinskiGerard.2004,HiroshimaMatte.2019}, and for the fibers of the translation-invariant model in \cite{Dam.2018,DamHinrichs.2021}. 
A similar result for the translation-invariant standard model of non-relativistic quantum electrodynamics, the Pauli--Fierz model, can be found in \cite{HaslerHerbst.2008a}.

Despite these indications for the absence of a ground state in infrared singular models, it has also been observed that
internal symmetries can lead to the cancellation of infrared divergences, e.g., for the Pauli--Fierz model at zero total momentum \cite{Chen.2008,HaslerSiebert.2020}.
In this article, we further investigate a specific model in which this phenomenon has been observed at small coupling -- the spin boson model.

\smallskip
The spin boson model describes the interaction of a two-state quantum system (called spin) with a bosonic quantum field.
It is a versatile model for various physical situations and thus extensively studied in the mathematics and physics literature alike.
Due to the richness of the available literature and studied properties, we refrain from giving specific references here, but refer to the reference sections of the articles cited below. In the present work we show that for certain forms of the spin-field-coupling, the model has no ground state for high coupling strength. Together with known results for the same model about the existence of a ground state for low coupling strength, this establishes a phase transition, which we show to be unique.

Significant progress in the understanding of the spin boson model was 
achieved in \cite{Spohn.1989} by Herbert Spohn. Although the paper is 
highly relevant for our work, and we even closely follow some of its key 
methods, the setting is slightly different. Spohn treats what is called 
the Ohmic case of the interaction, which is a special case of the types of 
interactions that we consider. He proposes to define ground states as 
zero temperature limits of thermal states. Since the Hamiltonian of the 
spin boson model is not of trace class, one has to use KMS states, 
which breaks the direct connection of thermal states and spectral theory.
In this framework, it is then shown that such 'thermal' ground states 
always exist, but a phase transition occurs in how they behave. Since the 
relations of Spohn's results to ours are non-trivial and interesting, we 
will give further details on them at the end of \cref{sec:Ising} 
below after having introduced the model and its functional integral 
representation.

While the results of \cite{Spohn.1989} certainly indicate that the spectrum of the spin boson Hamiltonian changes as the coupling strength varies, there appears to be no obvious direct relation to spectral theory. Therefore this question has been tackled independently, often by completely different methods. 
The existence of ground states in the spin boson model even in infrared-divergent situations was first proven by Hasler and Herbst \cite{HaslerHerbst.2010}, who used operator-theoretic renormalization and thus, due to its perturbative nature, could only provide a proof at small spin-field coupling.
The result was later reproved using iterated perturbation theory (also known as Pizzo's method) in \cite{BachBallesterosKoenenbergMenrath.2017}, again for small coupling.

A first non-perturbative study of the existence of ground states in the spin boson model was performed in the article series \cite{HaslerHinrichsSiebert.2021a,HaslerHinrichsSiebert.2021b,HaslerHinrichsSiebert.2021c},
using a compactness argument going back to \cite{GriesemerLiebLoss.2001}.
To be more precise, the main result of \cite{HaslerHinrichsSiebert.2021a} reduces the proof for existence of ground states to checking an energy derivative bound.
In \cite{HaslerHinrichsSiebert.2021c}, the energy derivative was then translated into the magnetic susceptibility of a continuum Ising model, for which an upper bound was proven in \cite{HaslerHinrichsSiebert.2021b}.	
Since it is well-known for the discrete Ising model that the magnetic susceptibility undergoes a phase transition \cite{FrohlichSpencer.1982,NewmanSchulman.1986,ImbrieNewman.1988}, the absence of ground states in the spin boson model with strong coupling was conjectured,
also see \cite{Hinrichs.2022,Hinrichs.2022b}. 

\smallskip

In this article, we prove this conjecture. In fact, in \cref{thm:GSoverlap} below, we give an explicit formula that expresses the 
$L^2$-overlap of the ground state with the ground state of the uncoupled Hamiltonian,
in terms of correlation functions of a one-dimensional, long range, continuum Ising model; 
absence of the ground state is  signalled by this formula yielding the value $0$. 

The fact that the functional integral representation of the 
spin boson model is related to a continuum Ising model has been first noticed in \cite{EmeryLuther.1974} and has been exploited for different purposes in \cite{SpohnDuemcke.1985,Spohn.1989,FannesNachtergaele.1988,Abdessalam.2011}.
In \cite{HirokawaHiroshimaLorinczi.2014}, it is used to show existence, uniqueness and properties of the ground state, but only in the infrared regular case. A significant technical difference from our work to all previous ones is that we obtain a one-sided long range Ising model, while other works either obtain two-sided models or, in the case of thermal states as in \cite{Spohn.1989}, periodic boundary conditions. Due to the long range interaction, these differences are potentially significant. 

Non-existence of the ground state is established by proving long range order in the 
above mentioned one-sided Ising model. This result is well-known for discrete two-sided long range Ising models, e.g., \cite{FrohlichSpencer.1982,AizenmanChayesChayesNewman.1988,ImbrieNewman.1988},
but it is non-obvious how to adapt this to our continuum Ising model.
Here we follow ideas present in  \cite{Spohn.1989}, modifying the arguments to account for a change of boundary conditions and domain; both changes are not a priori harmless for long range models. Our proof exploits a discretization scheme that is 
already present in \cite{SpohnDuemcke.1985} and is systematically used in  \cite{HaslerHinrichsSiebert.2021b}, as well as results from discrete percolation theory. These are further referenced below.

\smallskip
Let us also give an outlook to future research directions based on this article.
A rather obvious question is about a quantitative version of Theorems \ref{thm:SB.int} and \ref{thm:SB.IRdiv}, e.g.,\ an upper bound on the value of the coupling strength
$\lambda_0$ that separates existence of the ground state from non-existence.
Whereas our discretization approach would in principle allow for an explicit upper bound, by tracking the critical coupling in the discrete percolation models, this seems tedious and is hence omitted here, because already the original articles on the discrete models \cite{NewmanSchulman.1986} refrain from giving explicit upper bounds.
A better approach with the potential to obtain (at least asymptotically) optimal bounds would be to work directly in the continuum Ising model, e.g., by adapting the methods from \cite{FrohlichSpencer.1982}, or by studying continuum percolation models in more detail.

More generally, whereas with the results of \cite{HaslerHinrichsSiebert.2021a,HaslerHinrichsSiebert.2021c} and this article, sufficient criteria on the existence and absence of ground states in terms of Ising correlation functions are available,
the sharpness of these criteria remains open.
This problem can be reduced to the study of sharpness in the phase transition of our continuum Ising model,
i.e.,
whether the absence of long range order in the one-sided model immediately implies a finite susceptibility in the two-sided model.
Results on discrete Ising models in this direction are available \cite{AizenmanBarskyFernandez.1987,DuminilCopinTassion.2016},
but their generalization to continuum models of the type studied here (and the adjustment of the domain) is left to future research.

In the broader context of non-relativistic quantum field theories, phase transitions w.r.t. the existence of ground states are also an interesting object of investigation.
Especially, when one turns to translation-invariant polaron models such as the Fr\"ohlich polaron \cite{Frohlich.1954} or the Bose polaron \cite{GrusdtDemler.2016} at fixed total momentum $P$, then the existence of ground states at small total momentum is well-known \cite{Moller.2006b,Polzer.2023,HinrichsLampart.2023} and the absence of ground states at large total momentum is at least conjectured, but except for the ultraviolet regularized Fr\"ohlich polaron \cite{Moller.2006b} not proven.
Our results connecting the long range behaviour of Ising correlation functions and the existence of ground states bear the potential to shed some light on possible absence proofs in this context as well.

\subsection*{Acknowledgements}

The authors would like to thank Herbert Spohn for many helpful discussions and comments.
BH wants to thank David Hasler and Fumio Hiroshima for valuable discussions on the subject of this article.
We also thank the anonymous referees for detailed reading of our paper and very useful comments. 

\subsection*{Funding}
This work was supported by DFG Grant No. 535662048. BH acknowledges support by the Ministry of Culture and Science of the State of North Rhine-Westphalia within the project `PhoQC'. MNK acknowledges support by the Studienstiftung. 
SP acknowledges funding from the Swiss State Secretariat for Education, Research and Innovation (SERI) through the consolidator grant ProbQuant, and funding from the Swiss National Science Foundation through the NCCR SwissMAP grant. BH and SP further acknowledge funding within the DFG scientific network `A(E)MP' (Grant No. 569490025).

\subsection*{Structure of the Article}

In \cref{sec:SB}, we  define the spin boson model and state our main results, which are proven in the subsequent sections.
We then define a continuum Ising model and relate its $n$-point functions to the (non-)existence of ground states in \cref{sec:Ising}.
The main results are the vacuum overlap expansion \cref{thm:GSoverlap} as well as its corollary, \cref{cor:absence}, which states that long range order in the infrared divergent model implies the absence of ground states.
Long range order at large coupling for the continuum Ising model related to the infrared-divergent spin boson model is proven in \cref{sec:longrangeorder}.

\section{Ground States in the Spin Boson Model}
\label{sec:SB}
In this section, we present our main results on the spin boson model.
In \cref{sec:SB.Fock}, we introduce the necessary notation on bosonic Fock space.
Then, in \cref{sec:SB.def}, we define the spin boson Hamiltonian.
Our results on ground states of the model are then presented in \cref{subsec:existenceofGS}.
All of them are corollaries of the Ising model duality which we discuss in the subsequent \cref{sec:Ising}.
\subsection{Fock Space Calculus}
\label{sec:SB.Fock}
Let us briefly define and collect the relevant notions for bosonic Fock spaces over $L^2(\IR^d)$.
For more detailed introductions, we refer to the textbooks \cite{Parthasarathy.1992,Arai.2018}.
Let us remark that the choice of $L^2(\IR^d)$ in this context is not strictly necessary,
but since our results require the validity of the Feynman--Kac formula \cref{lem:overlap} which has not been discussed for other choices in the literature, the (presumably straightforward) generalization to spin boson models over different $L^2$-spaces is not treated here.

The bosonic Fock space is defined by
\begin{align*}
	\FS = \bigoplus_{n=0}^\infty\FS^{(n)}, \quad \FS^{(0)}=\IC, \quad n\geq1\!:\, \FS^{(n)}= (L^2(\IR^d))^{\otimes_{\sfs\sfy\sfm}^n} \cong L^2_{\sfs\sfy\sfm}(\IR^{d\cdot n}),
\end{align*}
where the symmetrization occurs over the $n$ $\R^d$-variables.

Given any operator $A$ on $L^2(\IR^d)$, we define its second quantization as the selfadjoint operator
\begin{align*}
	\dG(A) = \bigoplus_{n=0}^\infty \dG_n(A), \quad \dG_0(A) = 0, \quad n\geq1\!:\,\dG_n(A) = \sum_{k=1}^{n} \underbracket{\Id\otimes\cdots \Id}_{(k-1)\text{-times}} \otimes A\otimes \underbracket{\Id \otimes\cdots\otimes \Id}_{(n-k)\text{-times}}.
\end{align*}
Given $f\in L^2(\IR^d)$, we now define the corresponding field operator $\ph(f)$ as the selfadjoint operator associated to the closure of the quadratic form
\begin{align*}
	\sum_{n=0}^\infty 2\sqrt{n+1}\Re \int_{\IR^{d\cdot n}}\int_{\IR^d} \overline{\psi^{(n)}(k)f(p)}\psi^{(n+1)}(p,k)\d p\,\d k,
	\quad
	\psi\in \FS
	\text{ s.t. }
	\exists N\in\IN~\forall n\ge N\!: \psi^{(n)}=0\,;
\end{align*}
see for example \cite[\S\,5.9]{Arai.2018} for an operator-theoretic definition.

If $A\ge0$ is injective and $f\in\cD(A^{-1/2})$,
then $\dG(A)\ge0$
and the well-known relative bound
\begin{align}
	\label{eq:standardrelbound}
	\begin{aligned}
		&\norm{\ph(f)\psi}
		\le
		\norm{(1+A^{-1/2})f}\norm{(1+\dG(A))^{1/2}\psi}
		, &&\psi\in\sD(\dG(A)^{1/2})\subset\sD(\ph(f))\\
		&\le \eps\norm{\dG(A)\psi} +\big(\norm{(1+A^{-1/2})f} +  (2\eps)^{-1}\norm{(1+A^{-1/2})f}^2\big)\norm{\psi}
		, && \psi\in\sD(\dG(A)),\ \eps>0 \quad
	\end{aligned}
\end{align}
holds.
\subsection{The Spin Boson Model}
\label{sec:SB.def}
Using the notions from the previous \lcnamecref{sec:SB.Fock}, we  define the spin boson model on $\IC^2\otimes\FS$.
As parameters, it takes the dispersion relation of the bosons $\omega$, the form factor $v$ describing the spin-boson-interaction and a coupling constant $\lambda\in\IR$.
It is  given by the expression
\begin{align}
	\label{def:SB}
	H_\lambda \coloneqq \begin{pmatrix}2 & 0 \\ 0 & 0\end{pmatrix}\otimes \Id +  \Id \otimes \dG(\omega)  + \lambda\begin{pmatrix}
		0 & 1 \\ 1 & 0
	\end{pmatrix}\otimes \ph(v)
\end{align}
(where $\omega$ is identified with a multiplication operator).
By the Kato--Rellich theorem \cite[Thm.~X.12]{ReedSimon.1975} and \cref{eq:standardrelbound}, the so-defined operator $H_\lambda$ is selfadjoint and bounded from below under the following from now global assumption.
\begin{hyp}\label{hyp:SA}Throughout this article, we assume that
	\begin{enumhyp}
		\item $\omega:\IR^d\to[0,\infty)$ satisfies $\omega(k)=\omega(-k)>0$ for almost all $k\in\IR^d$,
		\item $v\in L^2(\IR^d)$, and $\omega^{-1/2}v\in L^2(\IR^d)$.
	\end{enumhyp}
\end{hyp}
\begin{ex}\label{ex}
	Typical physical examples satisfying these assumptions are given by the choices
	\begin{align*}
		\omega(k) = |k|, \quad v(k) = \chi(k)|k|^{-\delta}, \quad \delta<\tfrac12(d-1),
	\end{align*}
	where $\chi:\IR^d\to[0,1]$ is an ultraviolet cutoff function, which ensures the square-integrability of $v$.
	Especially, in $d=3$ dimensions, the most studied model is given by the choice $\delta=\tfrac12$, similar to the much studied Nelson and Pauli--Fierz models.
\end{ex}
\begin{rem}
	We remark that the ultraviolet problem has recently been studied in \cite{DamMoller.2018b,Lonigro.2022,LillLonigro.2024,HinrichsLampartValentinMartin.2024},
	providing no-go results as well as Nelson-type selfenergy renormalized Hamiltonians.
	Since the existence of ground states in quantum field theories is, however, only connected to the infrared problem, see for example \cite{BachmannDeckertPizzo.2012,DamHinrichs.2021}, we will keep an ultraviolet regularization in place throughout this article.
\end{rem}
\subsection{Existence of Ground States}
\label{subsec:existenceofGS}
We are interested in  whether $E_\lambda = \inf \sigma(H_\lambda)$ is an eigenvalue of $H_\lambda$. In this case, we say that $H_\lambda$ {\em has a ground state}.
\begin{rem}
	\label{rem:unique}
	Let $\Omega_\downarrow\coloneqq\ket\downarrow\otimes\Omega$ be the ground state of $H_0$, where $\ket\downarrow := (0,1)$ is the ground state of the free spin and the vacuum $\Omega$ is the ground state of the field $\dG(\omega)$.
	In the case that $H_\lambda$ has a ground state $\psi_\lambda$, it is well-known that the ground state eigenspace $\ker(H_\lambda-E_\lambda)$ is one-dimensional and that $\langle \psi_\lambda, \Omega_\downarrow \rangle \neq 0$, see for example \cite[Thm.~2.3]{HaslerHerbst.2010} or \cite[Cor. 2.1]{HirokawaHiroshimaLorinczi.2014} (combined with the Perron--Frobenius--Faris theorem).
\end{rem}
Elaborating on the above observation, let us define the vacuum overlap
\begin{align}
	\label{def:rho}
	\rho(\lambda)\coloneqq \braket{\Omega_\downarrow,P_{\lambda}(\{E_\lambda\})\Omega_\downarrow},
\end{align}
where $P_{\lambda}$ is the spectral measure of $H_\lambda$ and hence $P_{\lambda}(\{E_\lambda\})$ is the orthogonal projection onto $\ker(H_\lambda-E_\lambda)$.
Thus, the statement of \cref{rem:unique} can be summarized as
\begin{align}\label{eq:gsoverlap}
	\rho(\lambda)>0 \iff H_\lambda \text{ has a ground state.}
\end{align}
Let us now state our main results on the existence of ground states, all of which will be proven in the subsequent sections, by studying the function $\rho$.

Our first result is that existence of a ground state at finite coupling implies existence for all smaller couplings and absence of a ground state at finite coupling implies absence for all larger couplings; i.e., if the spin boson model undergoes a phase transition in this respect, then it is unique.
\begin{thm}\label{thm:SB.int}
	There exists $\lambda_0\in[0,\infty]$ such that $H_\lambda$ has a ground state for $|\lambda|<\lambda_0$ and no ground state for $|\lambda|>\lambda_0$.
\end{thm}
\begin{rem}\label{rem:lambda0}
	An interesting question is whether $H_{\lambda_0}$ has a ground state.
	We recall that Spohn
	\cite{Spohn.1989} was able to prove that the boson number expectation diverges directly at the critical coupling (which is not necessarily the same as ours, though the conjecture seems natural).
\end{rem}
We now consider $\lambda_0$ more closely, starting with a well-known result on infrared-regular spin boson models. In the context of 
\cref{ex} this means $\delta<\tfrac12d-1$.
\begin{thm}
	\label{thm:SB.IRreg} If $\frac v\omega\in L^2(\IR^d)$, then $\lambda_0=\infty$.
\end{thm}
\begin{rem}
	This was already previously well-established.
	Without claiming completeness here, proofs for $\lambda_0>0$ can be found in \cite{AraiHirokawa.1995} for the massive case $\essinf\omega>0$ and in \cite{BachFroehlichSigal.1998a} using operator-theoretic renormalization. Proofs for $\lambda_0=\infty$ can be found in \cite{Gerard.2000,DamMoller.2018a,Hinrichs.2022b} using functional-analytic methods and in \cite{HirokawaHiroshimaLorinczi.2014} using a path integral method close to ours.
\end{rem}
We move to the infrared-divergent case, i.e., the case $\delta\in [\tfrac12d -1,\tfrac12d -\tfrac12)$ for \cref{ex}. In general quantum field theory, one then expects the absence of ground states for any $\lambda\ne 0$. However, for the spin boson model, exploiting the so-called spin-flip symmetry, $\lambda_0>0$ was established in \cite{HaslerHerbst.2010}.
A second proof for this fact was later presented in \cite{BachBallesterosKoenenbergMenrath.2017},
and the non-perturbative method of the article series \cite{HaslerHinrichsSiebert.2021a,HaslerHinrichsSiebert.2021b,HaslerHinrichsSiebert.2021c} provided a first lower bound, which we also state here.
In fact, we will review the proof from the perspective of the analysis of Ising model correlation functions at the end of \cref{sec:Ising}.
\begin{prop}\label{prop:SB.IRdiv}
	If $\omega$ is continuous and the zeros of $\omega$ are isolated, then $\lambda_0\ge5^{-1/2}\norm{\omega^{-1/2}v}_2^{-1}$.
\end{prop}
The main result of this article is the opposite observation that $\lambda_0$ is finite.
We prove it at the end of \cref{sec:longrangeorder}.
\begin{thm}\label{thm:SB.IRdiv}
	If $\int_{\IR^d}\ee^{-t\omega(k)}\abs{v(k)}^2\d k\ge C(1+t^2)^{-1}$ for some $C>0$ and all $t>0$, then $\lambda_0<\infty$.
\end{thm}
If $\omega$ and $v$ are as in \cref{ex} then $
\int_{\IR^d}\ee^{-t\omega(k)}\abs{v(k)}^2\d k \sim t^{2\delta-d}$ as $t\to \infty$. In other words, in this case the assumption of \cref{thm:SB.IRdiv} is satisfied if and only if $\delta\in [\tfrac12d -1,\tfrac12d -\tfrac12)$, i.e., if and only if $v/\omega \notin L^2(\R^d)$.

Finally, let us note that our method also provides an upper bound for the vacuum overlap of the ground state, recall  \cref{def:rho} for the definition.

\begin{thm}\label{thm:overlapupperbound}
	We have
	\begin{align*}
		\rho(\lambda) \le \exp\left(-\lambda^2 \int_0^\infty\int_{\IR^d}t|v(k)|^2\ee^{-t(\omega(k)+2)} \d k\d t\right).
	\end{align*}
\end{thm}

\section{The Spin Boson Model as Continuum Ising Model}
\label{sec:Ising}

The key method in this article is the rewriting of the spin boson model defined in \cref{def:SB} into a continuous Ising model.
We will first introduce this class of continuum Ising models and then relate them to the functional integral representations of spin boson models in \cref{lem:overlap}.

Let $\mathbb P$ denote the law of the standard continuous time random walk on 
$\{-1,1\}$ with uniform initial distribution. In other words, the respective 
stochastic process $X = (X_t)_{t\geq 0}$ fulfils $\mathbb P(X_0=1) = \mathbb
P(X_0=-1) = 1/2$, and jumps from one element of $\{-1,1\}$ to the other occur after $\operatorname{Exp}(1)$-distributed times independent of each other as well as the starting distribution.
We perturb the process $(X_t)$ by a pair potential, similar to what happens in 
the theory of Gibbs measures. Let $g: \IR \to [0,\infty)$ be even and continuous, $\alpha \geq 0$, and set 
\begin{align}\label{def:Ising}
	\mathbb P_{\alpha, T}(\d X) \coloneqq \frac{1}{Z_{\alpha, T}}\exp\Big(\alpha \iint_{[0, T]^2} g(t-s) X_s X_t \, \d s \d t \Big) \,\mathbb P( \d X)
\end{align}
where the partition function $Z_{\alpha, T}$ acts as a normalization constant.
The nonnegative function $g$ introduces an attractive interaction, attributing a larger weight to pairs $(s,t) \in [0,T]^2$ where $X_s$ and $X_t$ have the 
same sign.
The measure can also be obtained as an appropriate scaling limit of discrete Ising models, see for example \cite{SpohnDuemcke.1985,HaslerHinrichsSiebert.2021b} as well as the proof of \cref{lem:GKS}.
After taking the limit, the nearest neighbour interaction is directly encoded in the underlying a-priori measure $\mathbb P$.
Therefore it is justified to refer to the process with path measure 
$\mathbb P_{\alpha, T}$ as a continuum Ising model.

While some of our results hold for general continuous and even functions $g$, the connection to the spin boson model appears through the choice 
\begin{align}
	\label{eq:SBIsingMapping}
	g(t) = \int_{\IR^d} |v(k)|^2\ee^{-|t| \omega(k)} \d k,
	\qquad
	\alpha = \alpha(\lambda)\coloneqq \frac{\lambda^2}{2},
\end{align}
where $v$, $\omega$ and $\lambda$ are as in \cref{sec:SB}; in particular, $v$ and $\omega$ satisfy \cref{hyp:SA}, and $\lambda$ is the coupling constant
of the spin boson model. The connection is via a Feynman-Kac formula. 
\begin{lem}\label{lem:overlap}
	If $g$ and $\alpha(\lambda)$ are given by \cref{eq:SBIsingMapping},
	then
	\begin{align*}
		\braket{\Omega_\downarrow,\ee^{-TH_\lambda}\Omega_\downarrow} = Z_{\alpha(\lambda),T}.
	\end{align*}
\end{lem}
\begin{proof}
	Under the additional assumption that $v$ has real Fourier transform, i.e., $v(k)=\overline{v(-k)}$, proofs can be found, e.g.,\ in  \cite{HirokawaHiroshimaLorinczi.2014,HaslerHinrichsSiebert.2021c}. 
	
	To obtain such a Feynman--Kac formula, it is necessary to diagonalize the interaction term in spin space.
	This is done by conjugating with the unitary mapping $\ket{\downarrow}= (0,1)$ to $\frac{1}{\sqrt{2}} (1,1)$, $\ket{\uparrow}= (1,0)$ to $\frac{1}{\sqrt{2}} (1,-1)$. Then testing with the vector $\Omega_\downarrow$ results in the free boundary conditions in $Z_{\alpha(\lambda),T}$, implemented by our choice of initial distribution, and not in periodic ones as used in \cite{Spohn.1989}, even when considering a diagonal element of the spin boson semigroup.
	
	To extend the result to arbitrary $v$, let $B$ be the unitary operator after setting $D(k))=|v(k)|^2+|v(-k)|^2$ given by
	\[ B\psi(k) = \frac{1}{\sqrt{2D(k)}}\Big((\overline{v(k)}-v(-k))\psi(k) + (v(k)+\overline{v(-k)})\psi(-k)\Big), \quad D(k)\ne 0 \]
	and $B\psi(k)=\psi(k)$ if $D(k)=0$.
	Now $Bv$ has real Fourier transform, since 
	\begin{align}
		\label{eq:Bv}
		Bv(k)=Bv(-k)=\sqrt{\tfrac12D(k)}.
	\end{align}
	Further, $B$ commutes with $\omega$, by the assumption that $\omega$ is even.
	Defining its unitary second quantization by
	\[ \Gamma(B) = \Id\oplus\bigoplus_{n=1}^\infty B^{\otimes n}, \]
	we thus find
	\begin{align*}
		\Gamma(B)H_\lambda\Gamma(B)^* = \begin{pmatrix}2 & 0 \\ 0 & 0\end{pmatrix}\otimes \Id +  \Id \otimes \dG(\omega)  + \lambda\begin{pmatrix}
			0 & 1 \\ 1 & 0
		\end{pmatrix}\otimes \ph(Bv).
	\end{align*}
	Since $Bv$ has real Fourier transform, $\Gamma(B)\Omega_\downarrow =\Omega_\downarrow$ and $g(t)$ as defined in \cref{eq:SBIsingMapping} does not change when $v$ is replaced by $Bv$, by \cref{eq:Bv} and the fact that $\omega$ is even, the claim follows.
\end{proof}

The key idea of our proof is now to express the vacuum overlap $\rho$ in terms of spin correlations in our continuum Ising model, with $g$ and $\alpha$ as given in \cref{eq:SBIsingMapping}. For this, we will make use of the identity
\begin{equation}
	\label{Equation: Rho and GC}
	\rho(\lambda) = \lim_{T\to \infty} \frac{Z_{\alpha(\lambda), T}^2}{Z_{\alpha(\lambda), 2T}}.
\end{equation} 
While representations of this type are not new, see e.g. \cite[Thms.~4.130,131]{LorincziHiroshimaBetz.2020} or \cite[Lemma~4.10]{BetzPolzer.2022}, we will show \cref{Equation: Rho and GC} in the proof of \cref{thm:GSoverlap} below. In order to upper bound $\rho(\lambda)$, we will express $1/\rho(\lambda)$ in terms of our model and lower bound the respective expression.
For $s\leq t$ let us denote by $\mathcal D(s, t)$ the set of jump points of $X$ in the interval $(s, t]$. Then
\begin{equation}
	\label{eq: potential in terms of jump process}
	X_s X_t = (-1)^{|\mathcal D(s, t)|}.
\end{equation}
Under the unperturbed measure $\mathbb P$, the jump points of $X$ form a Poisson point process whose intensity measure is the Lebesgue measure. In particular, for $(s, t)\in [0, T]^2, \, (u, v) \in [T, 2T]^2$, the random variables $|\mathcal D(s, t)|$ and $|\mathcal D(u, v)|$ are  independent under $\mathbb P$. Thus
\begin{equation}
	\label{Equation: Definition of P tilde}
	\widetilde{\mathbb P}_{\alpha, 2T}(\d X) = \frac{1}{Z_{\alpha, T}^2} \exp\Big(\alpha \iint_{[0, T]^2 \cup [T, 2T]^2} g(t-s) X_sX_t \, \d s \d t\Big) \, \mathbb P(\mathrm dX)
\end{equation} 
defines a probability measure, and we have
\begin{equation}
	\label{Equation: fraction of Z in terms of expected value}
	\frac{Z_{\alpha, 2T}}{Z_{\alpha, T}^2} = \widetilde{\mathbb E}_{\alpha, 2T} \Big[\exp\Big(2\alpha \iint_{[0, T]\times [T, 2T]} g(t-s) X_sX_t \, \d s \d t\Big)\Big]
\end{equation}
where the expected value $\widetilde{\mathbb E}_{\alpha, 2T}$ is taken with respect to $\widetilde{\mathbb P}_{\alpha, 2T}$.
Conditionally on $\{X_T = 1\}$, a path $(X_t)_{t \in [0, 2T]}$ is uniquely determined by its jumping times. Hence, as a consequence of \cref{eq: potential in terms of jump process}, the processes $(X_{T+t})_{t\in [0, T]}$ and $(X_{T-t})_{t\in [0, T]}$ are under $\widetilde{\mathbb P}_{\alpha, 2T}(\, \cdot \, |X_T = 1)$ independent and identically distributed with distribution $\mathbb P_{\alpha, T}(\, \cdot \, | X_0 = 1)$. An application of Jensen's inequality to lower bound \cref{Equation: fraction of Z in terms of expected value} is already sufficient to conclude that long range order in the infrared-divergent model excludes the existence of a ground state. By expanding the exponential in \cref{Equation: fraction of Z in terms of expected value} into a series, we will express $\rho$ in terms of correlation functions which allows us to additionally conclude the uniqueness of the phase transition.

For notational brevity,
for $n \in \mathbb N$, $T>0$ and $t_1, \hdots, t_n \in [0,T]$,
let us denote the correlation functions of the process with path measure given by \cref{def:Ising}, 
as
\begin{equation*}
	\tau_{\alpha, n, T}(t_1, \hdots, t_n) \coloneqq \mathbb E_{\alpha, T}\Big[\prod_{i=1}^n X_{t_i} \big| X_0 = 1\Big] = \mathbb E_{\alpha, T}\Big[\prod_{i=1}^n (X_{t_i}X_0)\Big]
\end{equation*}
where we used in the last equality that
\begin{equation*}
	\mathbb E_{\alpha, T}\Big[\prod_{i=1}^n (X_{t_i} X_0) \Big] = \frac{1}{2} \mathbb E_{\alpha, T}\Big[\prod_{i=1}^n X_{t_i}\big|X_0 = 1\Big] + \frac{1}{2} \mathbb E_{\alpha, T}\Big[\prod_{i=1}^n -X_{t_i}\big|-X_0 = 1\Big]
\end{equation*}
as well as that $X$ and $-X$ have the same distribution. Of crucial importance is the validity of the GKS inequalities, which in particular imply non-negativity of $\tau_{\alpha, n, T}$.

\begin{lem} \label{lem:GKS}
	For the continuum Ising model, we have 
	\begin{align}
		\label{eq:GKS}
		\mathbb E_{\alpha, T}\Big[\prod_{i=1}^n X_{s_i}\Big] \geq 0, \quad \mathbb E_{\alpha, T}\Big[\prod_{i=1}^n X_{s_i} \cdot \prod_{i=1}^m X_{t_i} \Big] \geq \mathbb E_{\alpha, T}\Big[\prod_{i=1}^n X_{s_i}\Big]  \mathbb E_{\alpha, T}\Big[\prod_{i=1}^m X_{t_i}\Big]
	\end{align}
	for any $n, m \in \N$ and $s_1, \hdots, s_n, t_1, \hdots, t_m \in [0, T]$.
\end{lem}
\begin{proof}
	Write $\delta(T, N)\coloneqq TN^{-1}$ for $N\in \mathbb N^+$ and set
	\begin{equation}
		\label{Equation: Definition of Gamma}
		\Lambda(T, N) \coloneqq \{k\delta(T, N): \, 0\leq k \leq N\}.
	\end{equation}
	We define the probability measure $\mathbb P_{\alpha, T, N}$ on $\{-1, 1\}^{\Lambda(T, N)}$ by
	\begin{align}\label{eq:discreteIsing}
		\mathbb P_{\alpha, T, N} (\{\sigma\}) \propto \exp\Big(\sum_{\substack{\{i, j\} \subset \Lambda(T, N)\\i\neq j}} \big(2\alpha \delta(T, N)^2 g(|i-j|)\chr_{i\nsim j} - \tfrac{1}{2}\log(\delta(T, N)) \chr_{i\sim j} \big)\sigma_i \sigma_j  \Big)
	\end{align}
	where we write $i\sim j$ in case that $|i-j| = \delta(T, N)$. Since $\mathbb 
	P_{\alpha, T, N}$ is the measure of a ferromagnetic Ising model, the GKS 
	inequalities hold (see e.g. \cite{FriedliVelenik.2017}), 
	i.e.,\ \cref{eq:GKS} holds when $\mathbb E_{\alpha, T}$ is replaced by $\mathbb 
	E_{\alpha, T,N}$ throughout. 
	As noticed in \cite{SpohnDuemcke.1985} and rigorously shown in \cite[Prop.~4.1]{HaslerHinrichsSiebert.2021b}, we have for any $t_1, \hdots, t_k\in [0, T]$
	\begin{equation}
		\label{eq:convdiscIsing}
		\mathbb E_{\alpha, T}[X_{t_1} \cdots X_{t_k}]  =  \lim_{N \to \infty} \mathbb E_{\alpha, T, N}[X_{i_N(t_1)} \cdots X_{i_N(t_k)}]
	\end{equation}	
	where $i_N(t_j) \coloneqq \sup\{i \in \Lambda(T, N):\, i \leq t_j\}$ for all $1\leq j\leq k$. The claim follows.    
\end{proof}

It is a standard observation that the GKS inequalities allow to prove the existence of the thermodynamic limit as well as monotonicity w.r.t.\ the coupling strength.

\begin{lem}
	\label{lem:correlationthermodynamic}
	For all $\alpha>0$, $n\in \N$ and $t_1, \hdots, t_n \geq 0$ the limit
	\begin{equation*}
		\tau_{\alpha, n}(t_1, \hdots, t_n) \coloneqq \lim_{T\to \infty} \tau_{\alpha, n, T}(t_1, \hdots, t_n)
	\end{equation*}
	exists and is monotone increasing in $\alpha$, and more generally w.r.t.\ pointwise order in $\alpha g$.
\end{lem}
\begin{proof}
	We will in fact prove the stronger statement that $\tau_{\alpha, n, T}(t_1, \hdots, t_n)$ is, for any $n\in \N$ and $t_1, \hdots, t_n\in [0, T_0]$,
	an increasing function of both $T \geq T_0$ and $\alpha\geq 0$ as well as $g$.
	
	In the discrete case, the monotonicity in coupling and domain is a well-known consequence of the GKS inequalities, see for example \cite[Exercise 3.31]{FriedliVelenik.2017}. In the continuum case, the monotonicity in $\alpha$ and $g$ hence follows directly by taking the scaling limit. For the monotonicity in $T$, it suffices to restrict ourselves to showing monotonicity on $\mathbb Q \cap [T_0, \infty)$, since an application of the dominated convergence theorem in the definition \cref{def:Ising} implies continuity of $T\mapsto \tau_{\alpha,n,T}(t_1,\ldots,t_n)$. Let $T_1, T_2\in \mathbb Q \cap [T_0, \infty)$ with $T_1<T_2$ and $p_1, p_2, q_1, q_2 \in \mathbb N$ such that $T_1 = p_1/q_1, T_2 = p_2/q_2$. Then for any $M\in \mathbb N^+$
	\begin{equation*}
		\Lambda(T_1, p_1 q_2 M) \subset \Lambda(T_2, p_2 q_1 M), \quad \delta(T_1, p_1 q_2 M) = \delta(T_2, p_2 q_1 M).
	\end{equation*}
	Hence, the monotonicity in $T$ follows from the monotonicity in the domain in the discrete model by taking $M\to \infty$.
\end{proof}

We are now ready to state and prove the main technical result of this work.
To our knowledge, an identity of this type has not appeared in the literature previously.
Below, the convolution  $f\ast h:[0, \infty)^n \to  [0, \infty)$ of two measurable functions $f, h:[0, \infty)^n \to [0, \infty)$ is defined as the ordinary convolution when extending $f$ and $h$ by zero to $\mathbb R^n$, i.e.,
\begin{equation*}
	(f\ast h)(t) \coloneqq \int_{[0, t]} f(s)h(t-s) \, \d s, \quad \text{ where } [0, t] \coloneqq \bigtimes_{i=1}^n [0, t_i].
\end{equation*}
\begin{thm}\label{thm:GSoverlap}
	Let $g$ and $\alpha(\lambda)$ be given by \cref{eq:SBIsingMapping}, and recall that $\rho(\lambda)$ is defined in \cref{def:rho}. We  have
	\begin{equation*}
		\rho(\lambda) = \bigg( \sum_{n=0}^\infty \frac{(2\alpha(\lambda))^n}{n!} \int_{[0, \infty)^n} (\tau_{\alpha(\lambda), n}\ast\tau_{\alpha(\lambda), n}) (t)  \prod_{i=1}^ng(t_i) \d t_i \bigg)^{-1}
	\end{equation*}
	where  $1/\infty \coloneqq 0$.
\end{thm}
\begin{proof}
	As a first step, we show that for all continuous, even functions $g: \mathbb R \to [0,
	\infty)$ and all $\alpha \geq 0$, 
	\begin{align}
		\label{eq:first step}
		\lim_{T \to \infty} \frac{Z_{\alpha, 2T}}{Z_{\alpha, T}^2} = 
		\sum_{n=0}^\infty \frac{(2\alpha)^n}{n!} \int_{[0, \infty)^n} (\tau_{\alpha, n}\ast\tau_{\alpha, n}) (t)  \prod_{i=1}^ng(t_i) \d t_i 
		\,\, \in [1, \infty],
	\end{align}
	where the existence of the limit is part of the claim. To show it, fix $T>0$ and recall the definition of $\widetilde{\mathbb P}_{\alpha, 2T}$ given in \cref{Equation: Definition of P tilde}.
	By expanding the exponential into a series, we obtain
	\begin{align*}
		\frac{Z_{\alpha, 2T}}{Z_{\alpha, T}^2} &= \widetilde{\mathbb E}_{\alpha, 2T}\Big[\exp\Big(2\alpha \int_0^T \d s \int_T^{2T} \d t \, g(t-s) X_s X_t\Big) \Big] \\
		&=\sum_{n=0}^\infty \frac{(2 \alpha)^n}{n!} \int_{[0, T]^n} \d s \int_{[T, 2T]^n} \d t \, \prod_{i=1}^{n} g(t_i-s_i) \widetilde{\mathbb E}_{\alpha, 2T}\Big[\prod_{i=1}^n X_{s_i} X_{t_i} \Big].
	\end{align*}
	As mentioned above, the processes $(X_{T+t})_{t\in [0, T]}$ and $(X_{T-t})_{t\in [0, T]}$ are under $\widetilde{\mathbb P}_{\alpha, 2T}(\, \cdot \, |X_T = 1)$ independent and identically distributed with distribution $\mathbb P_{\alpha, T}(\, \cdot \, | X_0 = 1)$. We hence get for all $s_1, \hdots, s_n \in [0, T]$ and $t_1, \hdots, t_n\in [T, 2T]$
	\begin{align*}
		\widetilde{\mathbb E}_{\alpha, 2T}\Big[\prod_{i=1}^n X_{s_i} X_{t_i} \Big] &= \widetilde{\mathbb E}_{\alpha, 2T}\Big[\prod_{i=1}^n X_{s_i} X_{t_i} \big| X_T = 1 \Big] \\
		&= \widetilde{\mathbb E}_{\alpha, 2T}\Big[\prod_{i=1}^n X_{s_i} \big| X_T = 1 \Big] \widetilde{\mathbb E}_{\alpha, 2T}\Big[\prod_{i=1}^n X_{t_i} \big| X_T = 1 \Big] \\
		&= \mathbb E_{\alpha, T}\Big[\prod_{i=1}^n X_{T-s_i} \big| X_0 = 1 \Big]\mathbb E_{\alpha, T}\Big[\prod_{i=1}^n X_{t_i-T} \big| X_0 = 1 \Big] \\
		&= \tau_{\alpha, n, T}(T-s_1, \hdots, T-s_n)\tau_{\alpha, n, T}(t_1-T, \hdots, t_n-T).
	\end{align*}
	Especially, $\tau_{\alpha,n,T}=0$ for $g=0$ and hence, as a consequence of \cref{lem:correlationthermodynamic}, $\tau_{\alpha, n, T}(u)$ is non-negative for every $n\in \mathbb N$ and $u\in [0, T]^n$. Moreover, the proof of \cref{lem:correlationthermodynamic} revealed that $\tau_{\alpha, n, T}(u)$ is increasing in $T$.
	A change of variables, monotone convergence, and a further change of variables hence give 
	\begin{align}
		\label{Equation: Representation of quotient of partition function in terms of correlations}		
		\frac{Z_{\alpha, 2T}}{Z_{\alpha, T}^2} & = \sum_{n=0}^\infty \frac{(2 \alpha)^n}{n!} \int_{[0, T]^n} \d u \int_{[0, T]^n} \d v \, \prod_{i=1}^{n} g(u_i+v_i) \tau_{\alpha, n, T}(u)\tau_{\alpha, n, T}(v)
		\\
		& \stackrel{T \to \infty}{\longrightarrow} \sum_{n=0}^\infty \frac{(2 \alpha)^n}{n!} \int_{[0, \infty)^n} \d u \int_{[0, \infty)^n} \d v \, \prod_{i=1}^{n} g(u_i+v_i) \tau_{\alpha, n}(u)\tau_{\alpha, n}(v) \nonumber \\
		&=\sum_{n=0}^\infty \frac{(2 \alpha)^n}{n!} \int_{[0, \infty)^n} \d t \prod_{i=1}^{n} g(t_i) \int_{[0, t]} \d s \, \tau_{\alpha, n}(t-s)\tau_{\alpha, n}(s), \nonumber
	\end{align}
	which is \cref{eq:first step}. 
	
	Now let $g$ and $\alpha(\lambda)$ be given as in \cref{eq:SBIsingMapping}.
	The second step of the proof is then provided by the identity for the vacuum overlap given in \cref{Equation: Rho and GC},
	%
	which we prove here for the convenience of the reader.
	First, we note that by our first step, the limit
	\begin{equation*}
		\tilde \rho(\lambda) \coloneqq \lim_{T\to \infty} \frac{Z_{\alpha(\lambda), T}^2}{Z_{\alpha(\lambda), 2T}}
	\end{equation*}
	exists in $[0,1]$. 	
	Combining the well-known identity
	\begin{equation*}
		E_\lambda = -\lim_{T\to \infty} \frac{1}{T} \log \langle \Omega_\downarrow, \ee^{-T H_\lambda} \Omega_\downarrow \rangle,
	\end{equation*}
	which follows similar to \cref{rem:unique}, see the references there and \cite[Theorem 2.2]{KellerLenzVogtWojciechowski.2015}, 
	%
	with
	\cref{lem:overlap} yields
	\begin{equation}
		\label{Equation: Ground state energy in terms of partition function}
		E_\lambda = -\lim_{T\to \infty} \log\big(Z_{\alpha(\lambda), T}\big)/T.
	\end{equation}
	Further
	\begin{equation}
		\label{Equation: Exact exponential growth of partition function}
		\rho(\lambda) = \lim_{T \to \infty} \ee^{T E_\lambda} Z_{\alpha(\lambda), T}
	\end{equation}
	similarly follows,
	by the spectral theorem. Hence, if $\rho(\lambda)>0$, then
	\begin{equation*}
		\rho(\lambda) = \lim_{T\to \infty} \frac{\big(Z_{\alpha(\lambda), T}\ee^{E_\lambda T}\big)^2}{Z_{\alpha(\lambda), 2T}\ee^{2E_\lambda T}} = 	\tilde \rho(\lambda).
	\end{equation*}
	It remains to prove that $\tilde \rho(\lambda) > 0$ implies that $\rho(\lambda) > 0$.
	As we saw in the proof of \cref{eq:first step}, $T \mapsto Z_{\alpha(\lambda),2T}/Z_{\alpha(\lambda),T}^2$ is increasing as a consequence of the GKS inequalities and \cref{Equation: Representation of quotient of partition function in terms of correlations} such that
	\begin{align*}
		Z_{\alpha(\lambda),2T}\le \frac{1}{\tilde\rho(\lambda)}Z_{\alpha(\lambda),T}^2
	\end{align*}
	for all $T\geq 0$. Iterating this yields for $T_k\coloneqq 2^k T_0$ with $T_0>0$, $k\in \mathbb N$ 
	\begin{align}\label{eq:decayupperbound}
		\ee^{T_kE_\lambda}Z_{\alpha(\lambda),T_k}
		\le \ee^{T_kE_\lambda} \tilde \rho(\lambda) \bigg(\frac{Z_{\alpha(\lambda),T_0}}{\tilde \rho(\lambda)} \bigg)^{2^k} = \tilde \rho(\lambda) \gamma(T_0)^{T_k/T_0}
	\end{align}
	where
	\begin{equation*}
		\gamma(T_0) \coloneqq \frac{\ee^{T_0 E_\lambda} Z_{\alpha(\lambda),T_0}}{\tilde\rho(\lambda)}.
	\end{equation*}
	Taking the logarithm in \cref{eq:decayupperbound}, dividing by $T_k$ and taking the limit $k\to \infty$ yields with \cref{Equation: Ground state energy in terms of partition function}
	\begin{equation*}
		0 = E_\lambda + \lim_{k \to \infty} \log \big(Z_{\alpha(\lambda), T_k}\big)/T_k \leq \log(\gamma(T_0))/T_0
	\end{equation*}
	and hence $\gamma(T_0) \geq 1$. Hence, $\rho(\lambda)/\tilde \rho(\lambda) = \lim_{T\to \infty} \gamma(T) \geq 1$ yielding the statement.
\end{proof}
Most of the results about the spin boson model are now corollaries of \cref{thm:GSoverlap}.
\begin{proof}[\textbf{Proof of \cref{thm:SB.int}}]
	\cref{thm:GSoverlap,lem:correlationthermodynamic} imply that
	$\lambda\mapsto \rho(\lambda)$ is even and decreasing in $|\lambda|$.
	Together with \cref{eq:gsoverlap}, this proves  \cref{thm:SB.int}. 
\end{proof}
Next, note that for all $\alpha \geq 0$, all $n$ and all $t = (t_1, \ldots, t_n)$ ,
\begin{align}\label{eq:correstimate}
	\big(\inf_{r \geq 0}\tau_{\alpha, 1}(r) \big)^{2n} \prod_{i=1}^n t_i \leq \int_{[0, t]} \prod_{i=1}^n \tau_{\alpha, 1}(s_i) \tau_{\alpha, 1}(t_i-s_i) \, \d s \leq \int_{[0, t]} \tau_{\alpha,n}(s) \tau_{\alpha, n}(t-s) \, \d s  \leq \prod_{i=1}^n t_i,
\end{align}
where the second inequality is due to the GKS inequalities \cref{eq:GKS}.
This gives the
\begin{proof}[\textbf{Proof of \cref{thm:SB.IRreg}}]
	By Fubini's theorem,
	\begin{align}
		\label{eq:Fubini}
		\frac v\omega\in L^2(\IR^d) \iff \int_0^\infty tg(t)\d t <\infty,
	\end{align}
	with $g$  given in 
	\cref{eq:SBIsingMapping}. The rightmost inequality of 
	\cref{eq:correstimate} together with \cref{thm:GSoverlap} now 
	proves the claim.
\end{proof}
This also gives us the upper bound in \cref{thm:overlapupperbound}.
\begin{proof}[\textbf{Proof of \cref{thm:overlapupperbound}}]
	We use the middle inequality in \cref{eq:correstimate} together with the inequality 
	\begin{equation*}
		\prod_{i=1}^n \tau_{\alpha, 1}(s_i) \tau_{\alpha, 1}(t_i-s_i) \geq  \prod_{i=1}^n \tau_{0, 1}(s_i) \tau_{0, 1}(t_i-s_i) = \prod_{i=1}^n \ee^{-2t_i}
	\end{equation*}
	that is due to \cref{lem:correlationthermodynamic}.
	Combining this with \cref{thm:GSoverlap} proves the statement.
\end{proof}
The leftmost inequality in \cref{eq:correstimate} gives the 
following step towards the proof of \cref{thm:SB.IRdiv}:
\begin{cor}\label{cor:absence}
	Assume that  $\frac v\omega \notin L^2(\IR^d)$. Assume further that the continuum Ising model exhibits long range order, i.e.,\ that $\inf_{t\geq 0} \tau_{\alpha(\lambda), 1}(t) > 0$. Then the 
	spin boson model does not have a ground state, i.e.,\ 
	$\rho(\lambda) = 0$.  
\end{cor}
\begin{proof}
	By \cref{eq:Fubini,eq:correstimate}, already the term corresponding to $n=1$ of the sum appearing in \cref{thm:GSoverlap} is infinite. This shows the claim. 
\end{proof}

The condition  $\frac v\omega \notin L^2(\IR^d)$ is clearly fulfilled under the assumptions of 
\cref{thm:SB.IRdiv}; what remains to be shown is thus that these assumptions also
imply long range order in the continuum Ising model.
We will do this in the next section, after discussing the proof of 
\cref{prop:SB.IRdiv} as the last item of the present section. 

For this, let us recall the following existence criterion, which essentially stems from the articles \cite{HaslerHinrichsSiebert.2021a,HaslerHinrichsSiebert.2021c}.
For a recent alternative proof based on the Feynman--Kac formula \cref{lem:overlap} and a generalization of \cref{Equation: Rho and GC} see \cite{HinrichsPolzer.2025}.
\begin{prop}\label{prop:existencecorbound}
	Assume that $\omega$ is continuous and the zeros of $\omega$ are isolated.
	If \cref{eq:SBIsingMapping} and
	\begin{align}
		\label{Equation: Condition for existence}
		\limsup_{T\to \infty}\frac1T\int_0^T\int_0^T\EE_{\alpha(\lambda),T}[X_tX_s]\d t\d s<\infty,
	\end{align}
	hold, then $\rho(\lambda)>0$.
\end{prop}
\begin{rem}
	It would be desirable to phrase the above integrability criterion in terms of $\tau_{\alpha(\lambda),1}$ instead of more general correlation functions. Using monotonicity of correlations in the domain, it is not difficult to see that the integrability of $t\mapsto \tau_{\alpha(\lambda),1}(t)$ is necessary for \cref{Equation: Condition for existence} to hold. However, it is unclear to the authors if the reverse implication holds. 
\end{rem}
\begin{proof}
	Let us first review the strategy of \cite{HaslerHinrichsSiebert.2021a}.
	In the article, the authors study the Hamiltonian
	\begin{align*}
		\tilde H_{\lambda,m}(\mu) \coloneqq H_\lambda + m \Id\otimes \dG(1) + \mu\begin{pmatrix}0 & 1 \\1 & 0	\end{pmatrix}\otimes \Id
	\end{align*}
	and its ground state energy $\tilde E_{\lambda,m}(\mu)\coloneqq \inf\sigma(\tilde H_{\lambda,m}(\mu))$.
	For $m>0$, the model has a spectral gap of size $m$, cf. \cite[App. D]{HaslerHinrichsSiebert.2021c} and thus $\tilde E_{\lambda,m}$ is analytic in $\mu$.
	Combining a compactness argument going back to \cite{GriesemerLiebLoss.2001},
	also see \cite[Sec.~3]{HaslerHinrichsSiebert.2023} for an abstract statement,
	and second order perturbation theory, the main result \cite[Thm.~2.8]{HaslerHinrichsSiebert.2021a} under our assumptions gives the implication
	\begin{align*}
		\limsup_{m\downarrow0}|\tilde E_{\lambda,m}''(0)|<\infty
		\Longrightarrow
		\rho(\lambda)>0.
	\end{align*}
	To verify the assumption, we use \cite[Cor.~2.14]{HaslerHinrichsSiebert.2021c}, which gives
	\begin{align*}
		\tilde E_{\lambda,m}''(0) 
		&= -\lim_{T\to\infty}\frac1T \tilde \EE_{\alpha(\lambda),T} \bigg[ \bigg(\int_0^T X_t \d t\bigg)^2
		\bigg] \\
		&=-\lim_{T\to\infty}\frac1T \int_0^T\int_0^T\tilde \EE_{\alpha(\lambda),T}[X_tX_s]\d t\d s \\
		&\ge - \limsup_{T\to\infty}\frac1T \int_0^T\int_0^T\EE_{\alpha(\lambda),T}[X_tX_s]\d t\d s,
	\end{align*}
	where $\tilde \EE_{\alpha,T}$ is defined similarly to \cref{def:Ising}, but with $g$ replaced by the expression from \cref{eq:SBIsingMapping} with $\omega$ replaced by $\omega+m$.
	The lower bound is thus a simple corollary of the monotonicity in interaction strength, cf. \cref{lem:correlationthermodynamic,lem:GKS}.
	Combining the two observations above yields the statement.
\end{proof}
This gives us the
\begin{proof}[\textbf{Proof of \cref{prop:SB.IRdiv}}]
	In \cite{HaslerHinrichsSiebert.2021b}, the authors prove that
	\begin{align*}
		\lim_{T\to \infty}\frac1T\int_0^T\int_0^T\EE_{\alpha(\lambda),T}[X_tX_s]\d t\d s<\infty
	\end{align*}
	if $\alpha(\lambda)\|g\|_1< \frac 1{5}$; also see \cite[Thm.~4.21]{Hinrichs.2022}. Further, $\|g\|_1=2\|\omega^{-1/2}v\|_2^2$, by \cref{eq:SBIsingMapping} and Fubini's theorem. Combined with \cref{eq:SBIsingMapping,prop:existencecorbound}, this proves the statement.
\end{proof}

We end this section by recapitulating the results \cite{Spohn.1989} of Spohn, comparing
his findings to ours, and pointing out problems that we consider to be still open. Spohn uses an operator-algebraic approach and studies the zero
temperature limit of KMS states. Stated in our notation\footnote{Spohn already uses the unitary transform that we mention at the beginning
	of the proof of Lemma \ref{lem:overlap} in the definition of the Hamiltonian, so his definition looks different.
	Also, we only investigate the case where the parameter $\varepsilon$ in equation (1.1) of \cite{Spohn.1989} equals 1.}, 
Spohn studies the spin boson model with
Hamiltonian $H_{\lambda, h} = H_\lambda + h \left( {\scriptsize \begin{matrix} 0 & 1 \\ 1 & 0 \end{matrix} }\right)$,
for $h \in \mathbb R$. By results from \cite{FannesNachtergaeleVerbeure.1988a}, a unique KMS state
$\omega_{\beta, \lambda, h}$ for $H_{\lambda, h}$ exists\footnote{The standard reference for KMS states is \cite{BratteliRobinson.1996}. For the reader not familiar with the theory, it
	is best to imagine $\omega_{\beta,\lambda,h}$ as a substitute to the classical thermal state, i.e.,\ the map $A \mapsto \operatorname{tr}(A \ee^{-\beta H_{\lambda,h}})/\operatorname{tr}(\ee^{-\beta H_{\lambda,h}})$ on the Weyl algebra, which is a suitable subset of the linear bounded operators. Since $\ee^{-\beta H_{\lambda,h}}$ is not trace class due to 
	the presence of the field, the (well-established) theory of KMS states becomes necessary.
}
for each strictly positive inverse temperature $\beta$. 

Spohn shows that the zero temperature limit
$\omega_{\lambda,h} = \lim_{\beta \to \infty} \omega_{\beta,\lambda,h}$ exists for each $h \neq 0$, and that both the limits
$\omega^+_\lambda = \lim_{h \searrow 0} \omega_{\lambda,h}$ and $\omega^-_\lambda = \lim_{h \nearrow 0} \omega_{\lambda,h}$ exist. 
All limits are understood in the weak-$\ast$ sense. Under conditions that correspond to those in our Theorem \ref{thm:SB.IRdiv},
he then shows that there exists a critical value $\lambda_{\rm c} \in (0,\infty)$ 
so that $\omega^+_\lambda = \omega^-_\lambda$ for $\lambda < \lambda_{\rm c}$, and 
$\omega^+_\lambda \neq \omega^-_\lambda$ for $\lambda > \lambda_{\rm c}$. Moreover, he shows that the expected number of bosons
$\omega^\pm_\lambda({\rm d} \Gamma(\Id))$ is finite in the first case, and infinite in the second.

To obtain his results, Spohn represents $\omega_{\beta,\lambda,h}(A)$, for various operators $A$, as the expectation of a continuum Ising model.
Examples of such operators $A$ are products of Weyl operators with Pauli matrices, or the number operator; see \cite[Proposition 4.1]{Spohn.1989}. Taking the zero temperature 
limit $\beta \to \infty$ then corresponds to the infinite volume limit in the Ising model, and a phase transition of the Ising model
then implies the results on the spin boson model. Note in particular \cite[Equ. (8.8)]{Spohn.1989}, which links the expected number
of bosons to the decay of correlations in the long range Ising model.

All of this appears rather similar to the approach in the present article, but there are important differences. First of all, since the KMS states are essentially 
traces, the relevant Feynman-Kac formula leads to an Ising model with {\em periodic boundary conditions};
see, e.g.,\ also \cite{BetzUeltschi.2009} for a similar effect in the context of Bose-Einstein condensation. In contrast, the Feynman-Kac formula 
describing the spectral properties leads to the probability measure \eqref{def:Ising}, which has free boundary conditions, as we have discussed in the proof of \cref{lem:overlap}.
Because of the long range of the interaction, standard tools for comparing boundary conditions fail, and it is not clear (to the authors) whether 
the critical points of the two models coincide. Moreover, our arguments require us to {\em not} only understand 
correlations of the long range Ising model deep in the bulk, but also between the bulk and the boundary. The consequence is that 
we currently do not know whether $\lambda_{\rm c}$ (derived in \cite{Spohn.1989} and
called $\alpha_{\rm c}(\varepsilon)$ there) and our $\lambda_0$ coincide, although we do consider it likely to be the case. 

Another gap is the discrepancy between the finite (or infinite) expected number of bosons in Spohn's zero temperature ground states,
and the existence (or non-existence) of a spectral ground state in Fock space. In \cite[Remark 4]{Spohn.1989} it is stated that the 
`ground state lies in Fock space' for coupling constants below the critical value $\lambda_\rmc$. However, this 
just seems to refer to the fact that plugging the number operator into the limiting KMS state gives a finite result, and we are not aware of a 
theory that would imply that in such cases, a spectral ground state must exist. The fact that more recent papers (starting with 
\cite{HaslerHerbst.2010}, see our introduction) explicitly prove the existence of ground states by different methods shows that 
such a theory, if it exists, is certainly not widely known. 

In the case of large coupling constants, the gap is even wider. Spohn shows that the {\em expected number of bosons} in each of his
KMS states is infinite. Even if we assume that a theory that connects Spohn's ground states with spectral ground states exists, this 
only shows that such ground states are not in the domain of the number operator. They might well be in Fock space. On the other hand, 
the non-uniqueness of Spohn's ground state is a strong indication 
that no ground state exists in Fock space, since such a ground state must be unique by the Perron--Frobenius arguments recalled in \cref{rem:unique}.

In conclusion, although there are many intriguing connections between the work \cite{Spohn.1989} and ours (and ours certainly owes a lot to the ideas
presented there), significant gaps remain open, and it would be quite interesting to understand more about these issues.

\section{Long Range Order in Continuum Ising Models}
\label{sec:longrangeorder}\label{Section: Long-range order at large coupling}
To conclude the proof of \cref{thm:SB.IRdiv} from \cref{cor:absence},
we need to prove long range order under the given assumptions, which is the objective of this \lcnamecref{sec:longrangeorder}.
\begin{thm}
	\label{thm:LRO}
	If $g(t)\ge C(1+t^2)^{-1}$ for some $C>0$ and all $t\in\IR$, then for sufficiently large $\alpha>0$
	\begin{equation*}
		\inf_{t \geq 0} \tau_{\alpha, 1}(t) > 0.
	\end{equation*}
\end{thm}
\begin{rem}
	We remark that our slow decay assumption is always satisfied if our continuous function $g$ is strictly positive everywhere and $g(t)\sim  |t|^{-\kappa}$ for some $\kappa\le 2$ as $|t|\to\infty$.
\end{rem}
For the proof we modify an argument going back to Spohn \cite{Spohn.1989} (who showed non-zero magnetization under periodic boundary conditions at large coupling in the middle of the system) to adjust for the differences in boundary conditions (free rather than periodic; compare to the discussion inside the proof of \cref{lem:overlap}), domain ($[0, \infty)$ rather than $\R$), and position of interest (at the boundary rather than in the middle of the system).

\begin{figure}
	\centering
	\includegraphics[width=0.65\linewidth]{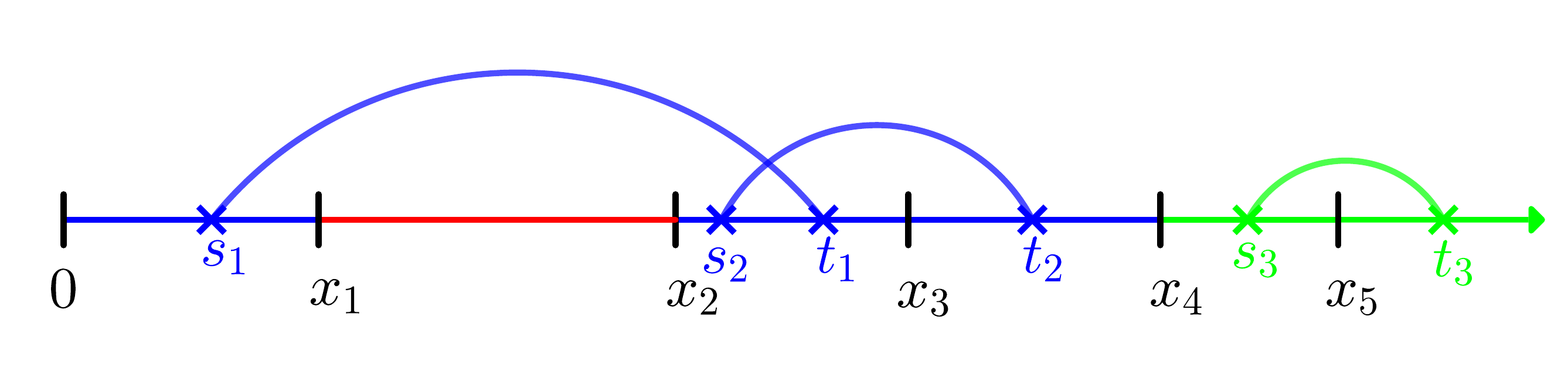}
	\caption{Visualization of the continuum percolation used in the proof of \cref{thm:LRO}. The connected components are depicted in blue, red and green. }
	\label{fig:sketchcontinuumpercolation}
\end{figure}

Let us define a continuum percolation on $[0, \infty)$ in the following manner:
\begin{itemize}
	\item Draw from a Poisson point process $\xi_1$ on $[0, \infty)$ that has intensity measure $ 2\d x$. Denote the points of $\xi_1$ by $x_1 < x_2 < \hdots$ and set $x_0 \coloneqq 0$.
	\item Independently from this, draw from a Poisson point process $\xi_2$ on $[0, \infty)^2$ with intensity measure $2 \alpha g(t-s) \chr_{\{s<t\}} \d s \d t$. We will call the points of $\xi_2$ bonds.
	\item For $n, m\in \N_0$ with $n<m$, connect the intervals $[x_n, x_{n+1})$ and $[x_m, x_{m+1})$ if and only if there exists a bond $(s, t)$ in $\xi_2$ such that $s \in [x_n, x_{n+1})$ and $t \in [x_m, x_{m+1})$.
\end{itemize}
For $x, y \in[0, \infty)$, we write $x \leftrightarrow y$ in case that $x$ and $y$ are contained in the same connected component, see  \cref{fig:sketchcontinuumpercolation} for a visualization. We denote by $\mathbf P_{\alpha,\sfc}$ the law of our continuum percolation.

The next \cref{Lemma: Lower bound on correlation via percolation} will allow us to prove long range order, by studying the continuum percolation. Essentially, the proof relies on a discretization argument for the percolation model, similar to the one described in the proof of \cref{lem:GKS} and discussed in \cref{Appendix: Convergence to the continuum percolation} in detail, as well as stochastic domination for discrete percolation models, cf. \cite{AizenmanChayesChayesNewman.1988}.
\begin{lem}
	\label{Lemma: Lower bound on correlation via percolation}
	For all $n \in \mathbb N$ we have $\tau_{\alpha, 1}(n) \geq \mathbf P_{\alpha,\sfc}(0 \leftrightarrow n)$. 
\end{lem}
\begin{proof}
	Our proof follows \cite[Section 5]{Spohn.1989}.
	We present it here, nevertheless, for the convenience of the reader.
	
	Let us fix some $T, N\in \mathbb N$ and $n\in \mathbb N \cap \Lambda(T, N)$ (recall \cref{Equation: Definition of Gamma}). As a first step, we note that we can identify our discrete Ising model with a FK-Percolation model on the complete graph $(\Lambda(T, N), \binom{\Lambda(T, N)}{2})$: We define the probability of a subgraph $\omega = (\Lambda(T, N), \mathcal E(\omega))$ by
	\begin{equation*}
		\mathcal P_{\alpha, T, N}(\{\omega\}) \propto 2^{\mathcal C(\omega)} \prod_{\{i, j\} \in \mathcal E(\omega)} \tilde p_{\alpha, N}(i, j)  \prod_{\{i, j\} \notin \mathcal E(\omega) } (1 - \tilde p_{\alpha, N}(i, j))
	\end{equation*}
	where $\mathcal C(\omega)$ is the number of connected components of $\omega$ and
	\begin{equation}
		\label{def:tildepaN}
		\tilde p_{\alpha, N}(i, j) = \begin{cases}
			1 - \delta(T, N) &\text{ if } |i - j| = \delta(T, N) \\
			1 - \ee^{- 4\alpha \delta(T, N)^2 g(i-j)} &\text{ else.}
		\end{cases}
	\end{equation}
	Then,   by \cite[Lemma 2.1]{AizenmanChayesChayesNewman.1988}, it holds that
	\begin{equation}
		\label{Equation: Correlations in terms of FK-Percolation}
		\mathbb E_{\alpha, T, N}(X_0 X_n) = \mathcal P_{\alpha, T, N}(0 \leftrightarrow n),
	\end{equation}
	where we recall that the left hand side was defined in \cref{eq:discreteIsing}.
	
	As a consequence of the FKG-inequality for FK-percolation, we have the stochastic domination \cite[Theorem 4.1]{AizenmanChayesChayesNewman.1988}
	\begin{equation}
		\label{eq:stochdom}
		\mathcal P_{\alpha, T, N} \succeq \mathbf P_{\alpha, T, N} 
	\end{equation} 
	where $\mathbf P_{\alpha, T, N}$ denotes independent bond percolation on $(\Lambda(T, N), \binom{\Lambda(T, N)}{2})$, where each edge $\{i, j\}$ is open with probability 
	\begin{equation}
		\label{def:paN}
		p_{\alpha, N}(i, j) \coloneqq
		\begin{cases}
			1 - 2\delta(T, N) &\text{ if } |i - j| = \delta(T, N) \\
			1 - \ee^{- 2\alpha \delta(T, N)^2 g(i-j)} &\text{ else.}
		\end{cases}
	\end{equation}
	As we show in \cref{Appendix: Convergence to the continuum percolation}, cf. \cref{Proposition: Convergence to continuum percolation},
	\begin{equation}
		\label{eq:percolim}
		\mathbf P_{\alpha, T, TN}(0 \leftrightarrow n) \xrightarrow{N\to\infty} \mathbf P_{\alpha, T, \sfc}(0 \leftrightarrow n),
	\end{equation}
	where $\mathbf P_{\alpha, T, \sfc}$ denotes the restriction of $\mathbf P_{\alpha, \sfc}$ to $[0, T]$.
	
	Combining \cref{eq:convdiscIsing,Equation: Correlations in terms of FK-Percolation,eq:stochdom,eq:percolim}
	and taking the limit $T\to \infty$ yields the claim by continuity from below.
\end{proof}
We can now move to the
\begin{proof}[\textbf{Proof of \cref{thm:LRO}}] 
	First, we notice that it is sufficient to show that $\inf_{n \in \mathbb N} \tau_{\alpha, 1}(n) > 0$, since
	the GKS inequalities \cref{lem:GKS} and the monotonicity \cref{lem:correlationthermodynamic} imply
	\begin{equation*}
		\tau_{\alpha, 1}(t) = \lim_{T\to \infty} \mathbb E_{\alpha, T}[X_0 X_{\lfloor t \rfloor} X_{\lfloor t \rfloor} X_t] \geq \lim_{T\to \infty} \mathbb E_{\alpha, T}[X_0 X_{\lfloor t \rfloor}]  \mathbb E_{0, T}[X_{\lfloor t \rfloor} X_t] \geq \tau_{\alpha, 1}(\lfloor t \rfloor) \ee^{-2}.
	\end{equation*}
	It is thus natural to work with discrete percolation models.
	
	\begin{figure}
		\centering
		\includegraphics[width=0.6\linewidth]{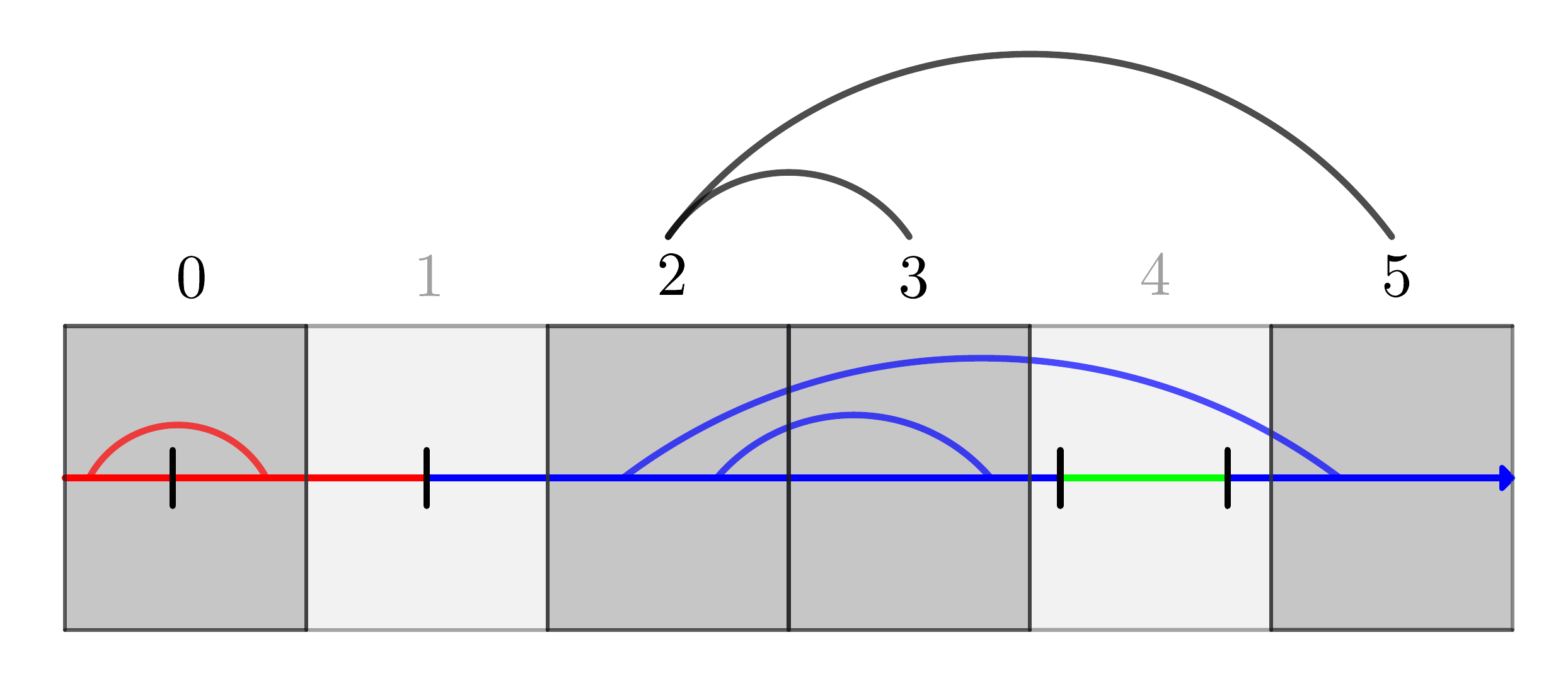}
		\caption{Discretization of the continuum percolation applied in the proof of \cref{thm:LRO}. An interval $[n, n+1]$ with $n=0, \hdots, 5$ is depicted in dark grey if it is alive and in light grey else.}
		\label{fig:discretizationcontinuumpercolation}
	\end{figure}
	
	Following the proof of \cite[Thm.~2\,(iii)]{Spohn.1989}, we define a site-bond-percolation on the complete graph $(\mathbb N, \binom{\mathbb N}{2})$: Fix some $n\in \mathbb N$. The process $\xi_1$ partitions the interval $[n, n+1]$ into subintervals. We call $n$ alive if all of these subintervals are connected via a path of bonds in $\xi_2$ that all start and end in $[n, n+1]$. We denote by $p_0(\alpha)$ the probability that $0$ is alive. For $n, m \in \N$ with $n< m$, we define the edge $\{n, m\}$ as open if there exists a bond $(s, t)$ of $\xi_2$ with $s\in [n, n+1)$ and $t\in [m, m+1)$. Then the probability that the edge $\{n, m\}$ is open is given by
	\begin{equation}
		\label{def:pnma}
		p_{n, m}(\alpha) = 1 - \exp\Big(-2\alpha \int_n^{n+1} \int_m^{m+1} g(t-s) \, \d s \d t  \Big).
	\end{equation}
	We denote by $\mathbf P_{\alpha, \sfd}$ the law of our discrete site-bond-percolation, see \cref{fig:discretizationcontinuumpercolation} for a visualization. For $t\in \mathbb N$, we write $0\leftrightarrow t$ in case that $0$ and $t$ are connected by a path of open edges that only contains vertices that are alive. Notice that for all $n \in \mathbb N$
	\begin{equation}
		\label{Equation: two point function percolation discrete vs continuous}
		\mathbf P_{\alpha, \sfc}(0 \leftrightarrow n) \geq \mathbf P_{\alpha, \sfd}(0 \leftrightarrow n).
	\end{equation}
	Since the results we want to apply utilize translation-invariance, we now extend the discrete percolation to all of $\IZ$.
	Therefore, for $\alpha>0$, we define $\mathbf Q_{\alpha, \sfd}$ as the equivalent of $\mathbf P_{\alpha,  \sfd}$ defined on all of $\mathbb Z$, i.e., the site-bond percolation on $(\mathbb Z, \binom{\mathbb Z}{2})$ where each $n\in \mathbb Z$ is alive with probability $p_0(\alpha)$ and where $\{n, m\}$ is open with probability $p_{n, m}(\alpha)$.
	
	Using that
	\begin{equation}
		\label{Equation: Estimate exponential}
		x - \frac{x^2}{2} \leq 1 - \ee^{-x} \leq x + \frac{x^2}{2}, \qquad x\ge 0,
	\end{equation}
	it follows that for fixed $\alpha>0$
	\begin{equation*}
		p_{n, m}(\alpha) \ge \frac{C\alpha}{2} |n-m|^{-2}
		\quad\text{for $|n-m|$ sufficiently large.}
	\end{equation*}
	One further easily observes that $p_0(\alpha), p_{0, 1}(\alpha) \uparrow 1$ as $\alpha \to \infty$.
	By \cite{NewmanSchulman.1986}, there hence\footnote{Notice that $p_{n, m}(\alpha)$ is increasing in $\alpha$. We can hence fix $p_{n, m}(\alpha)$ for $|n-m| \neq 1$ to be $p_{n, m}(\alpha_0)$ with $\alpha_0$ large and only increase $p_0(\alpha)$ as well as $p_{0,1}(\alpha)$.} exists a $\beta > 0$ such that
	\begin{equation}
		\label{Equation: System percolates}
		\mathbf Q_{\beta, \sfd}(0 \leftrightarrow \infty) > 0
	\end{equation}
	where $\{0 \leftrightarrow \infty \}$ denotes the event that the connected component of $0$ (considering only living sites) is infinite. As $\mathbf Q_{\beta, \sfd}$ defines a translation-invariant, irreducible site-bond percolation model, \cref{Equation: System percolates} implies that almost surely there exists a unique infinite connected component (see the discussion at the end of \cite{AizenmanKestenNewman.1987}), which we call $\mathcal C_\infty$.
	Hence,
	\begin{equation}
		\label{eq:infinitecluster}
		\mathbf Q_{\beta, \sfd}(0 \leftrightarrow t) \geq \mathbf Q_{\beta, \sfd}(0 \in \mathcal C_\infty, t \in \mathcal C_\infty) \geq \mathbf Q_{\beta, \sfd}(0 \in \mathcal C_\infty)^2 > 0,
	\end{equation}
	by the FKG inequality for site-bond percolation.
	
	It remains to prove
	\begin{equation}
		\label{eq:onesidedvstwosided}
		\mathbf P_{\alpha, \sfd}(0 \leftrightarrow t) \geq \mathbf Q_{\beta, \sfd}(0 \in \mathcal C_\infty)^2 > 0
	\end{equation}
	for $\alpha$ sufficiently large.
	W.l.o.g. and thanks to the monotonicity statement of \cref{lem:correlationthermodynamic}, we can restrict to the case $g(t)=\frac{C}{1+t^2}$.
	To compare our one-sided model, i.e., the site-bond percolation on $\mathbb N$ with law $\mathbf P_{\alpha, \sfd}$ with the two-sided model on $\Z$ with law $\mathbf Q_{\beta, \sfd}$,
	we define the bijection
	\begin{equation*}
		\varphi : \mathbb N_0 \to \mathbb Z, \quad \varphi(n) \coloneqq 
		\begin{cases}
			n/2 &\text{ if }n\in 2\mathbb N_0 \\
			-(n + 1)/2 &\text{ if }n\in 2\mathbb N_0 + 1
		\end{cases}
	\end{equation*}
	and extend $\varphi$ to a bijection $[0, \infty) \to \mathbb R\setminus (-1, 0)$ that is linear with slope $1$ on all intervals of the form $[2n, 2n+1)$ and  linear with slope $-1$ on all intervals of the form $[2n+1, 2n+2)$.
	Then $|\varphi(t) - \varphi(s)| \geq |t-s|/4$ for all $s, t\geq 0$ with $|\lfloor s \rfloor- \lfloor t \rfloor|\neq 2$.
	Inserting this together with our choice of $g$ into \cref{def:pnma},
	%
	we easily find $c>1$ such that
	\begin{equation*}
		\label{Equation: Estimate probabilities}
		p_{\varphi(n), \varphi(m)}(\alpha/c) \leq p_{n, m}(\alpha) \text{ for } |n-m| \notin \{0, 2\}.
	\end{equation*}
	Thus, if $\alpha>0$ is large enough that $\alpha/c > \beta$ and $p_{0, 2}(\alpha) > p_{0, 1}(\beta)$,
	then $p_{n, m}(\alpha) \geq p_{\varphi(n), \varphi(m)}(\beta)$ for all $n, m\in \mathbb N$ with $n\neq m$ and hence
	\begin{equation*}
		\mathbf P_{\alpha, \sfd}(0 \leftrightarrow t) \geq  \mathbf Q_{\beta, \sfd}(0 \leftrightarrow \varphi(t)). 
	\end{equation*}
	This combined with \cref{eq:infinitecluster} proves \cref{eq:onesidedvstwosided}.
	
	The statement now follows by combining \cref{Equation: two point function percolation discrete vs continuous,eq:onesidedvstwosided,Lemma: Lower bound on correlation via percolation}.
\end{proof}
Let us conclude by summarizing the
\begin{proof}[\textbf{Proof of \cref{thm:SB.IRdiv}}]
	This now follows by combining \cref{cor:absence,thm:LRO}.
\end{proof}

\appendix
\section{Convergence to the continuum percolation}
\label{Appendix: Convergence to the continuum percolation}
In the following, we will show the convergence of the discrete to the continuum percolation, as we used it in the proof of \cref{Lemma: Lower bound on correlation via percolation}, cf. \cref{eq:percolim}.
While an almost identical statement was already used in \cite{Spohn.1989}, there was no proof given therein.

To show convergence, we will identify our discrete percolation with a suitable point process. For $d\geq 1$ and a compact set $K \subset \mathbb R^d$, we denote by $\mathbf N(K)$ the configuration space of point processes on $K$,
i.e., the set of all integer-valued measures on $K$
\begin{equation*}
	\mathbf N(K) \coloneqq \{\mu \in \mathbf M(K):\, \mu(A) \in \mathbb N_0 \text{ for all }A\in \mathcal B(K)\},
\end{equation*}
where $\mathbf M(K)$ denotes the set of all finite Borel measures on $K$.

We define percolation on $\mathbf N([0, T]) \times \mathbf N([0, T]^2)$ in accordance with \cref{Section: Long-range order at large coupling}: Let $(\mu_1, \mu_2) \in \mathbf N([0, T]) \times \mathbf N([0, T]^2)$. The atoms $0\leq x_1 < \hdots < x_n \leq T$ of $\mu_1$ partition the interval $[0, T]$ into disjoint intervals 
\begin{equation*}
	[0, x_1) \cup [x_1, x_2) \cup \hdots \cup [x_{n}, T].
\end{equation*}
We connect two distinct intervals $I$ and $J$ of the partition if $\mu_2$ has an atom $(s, t)$ such that $s\in I$ and $t\in J$.

Our discrete percolation with law $\mathbf P_{\alpha, T, N}$ can be embedded by defining the point processes $(\xi_{N, 1}, \xi_{N, 2})$ by
\begin{equation*}
	\xi_{N, 1} = \sum_{ i\in \Lambda(T, N)\setminus \{0\}} \chr_{\{(i-\delta(T, N), i) \text{ is closed}\}} \delta_{i}, \quad \xi_{N, 2} = \sum_{ i, j \in \Lambda(T, N)} \chr_{\{j-i>\delta(T, N) \text{ and }(i, j) \text{ is open}\}} \delta_{(i, j)}.
\end{equation*}
Notice that the partitioning of $\Lambda(T, N)$ that is induced by $\xi_{N, 1}$ consists of the connected components one obtains by only considering nearest neighbour bonds.

Our continuum percolation $\mathbf P_{\alpha, T, \sfc}$ is given by the point processes $\xi = (\xi_{1}, \xi_{2})$:
Let $\xi_1$, $\xi_2$ be two independent Poisson point processes on $[0, T]$ and $[0, T]^2$ respectively, where $\xi_1$ has intensity measure $2\d x$ and $\xi_2$ has intensity measure $2\alpha g(t-s) \chr_{\{0\leq s<t\leq T\}} \, \d s \d t$.

We will show that $(\xi_{N, 1}, \xi_{N, 2})$ converges in distribution to $\xi = (\xi_{1}, \xi_{2})$ as $N\to \infty$. Showing that $\{0 \leftrightarrow n\}$ is a continuity set of the distribution $\mathbf P_{\alpha, T, \sfc}$ of $\xi$ will then imply \cref{eq:percolim} by the Portmanteau theorem \cite[Thm.~13.16]{Klenke.2020}.

\bigskip
Before continuing, let us recall a few facts about the convergence of random measures.
Therefore, we equip $\mathbf M(K)$ with the topology $\tau_w$ of weak convergence, i.e., the topology generated by the family of maps $\mu \mapsto \mu(f)$, $f\in C(K)$. It is well known that the space $(\mathbf M(K), \tau_w)$ is Polish and that its topology is generated by the Prohorov metric defined by
\begin{equation}
	\label{def:prohorov}
	d(\mu, \nu) \coloneqq \inf \{\varepsilon>0:\, \mu(A) \leq \nu(A_\varepsilon) + \varepsilon \text { and }\nu(A) \leq \mu(A_\varepsilon) + \varepsilon \text{ for all }A \in \mathcal B(K)\},
\end{equation}
where $A_\varepsilon \coloneqq \{y \in K:\, \operatorname{dist}(y, A) < \varepsilon\}$ for $A\in \mathcal B(K)$ and $\varepsilon>0$, see for example \cite{Kallenberg.2017}.
We denote by $\mathcal M(K) \coloneqq \sigma(\tau_w)$ the Borel-$\sigma$-algebra generated by $\tau_w$. As one can show with the monotone class theorem, $\mathcal M(K)$ coincides with the $\sigma$-algebra generated by the evaluation maps $\mu \mapsto \mu(A)$ with $A\in \mathcal B(K)$.
A random measure $\zeta$ on $K$ now is a $(\mathbf M(K), \mathcal M(K))$-valued random variable and a measurable set $A\in \mathcal B(K)$ is called a $\zeta$-continuity set if $\zeta(\partial A) = 0$ almost surely.

We define 
\begin{equation*}
	\mathcal G(K) \coloneqq \{(a, b] \cap K: a, b\in \mathbb R^d\}, \quad \mathcal S(K) \coloneqq \Big\{ \bigcup_{i=1}^n I_i:\, n\in \mathbb N,\, I_1, \hdots, I_n \in \mathcal G(K) \Big\}
\end{equation*}
(where $(a, b] \coloneqq \bigtimes_{i=1}^d (a_i, b_i]$ for $a, b \in \R^d$). Then $\mathcal S(K)$ is a DC-semi-ring in the terminology of \cite{Kallenberg.1973} and we obtain as a corollary of \cite[Theorem 1.1]{Kallenberg.1973}
\begin{lem}
	\label{Lemma: Fidi convegene}
	Let $\zeta$ be a random measure on $K$ such that all $S\in \mathcal S(K)$ are $\zeta$-continuity sets. Let $(\zeta_n)_n$ be a sequence of random measures on $K$. Then $\zeta_n \to \zeta$ in distribution if and only if for all $k\in \N$ and all pairwise disjoint $S_1, \hdots, S_k \in \mathcal G(K)$
	\begin{equation}
		\label{Equation: fidi convergence}
		(\zeta_n(S_1), \hdots, \zeta_n(S_k)) \to (\zeta(S_1), \hdots, \zeta(S_k))
	\end{equation}
	in distribution as $n\to \infty$.
\end{lem}
\begin{proof}
	By \cite[Theorem~1.1]{Kallenberg.1973}, $\zeta_n \to \zeta$ in distribution is equivalent to \cref{Equation: fidi convergence} holding for all $k\in \mathbb N$ and all (not necessarily pairwise disjoint) $S_1, \hdots, S_k \in \mathcal S(K)$. However, if $S_1, \hdots, S_k \in \mathcal S(K)$, then there exist pairwise disjoint $\tilde S_1, \hdots, \tilde S_m \in \mathcal G(K)$ and a linear map $f:\mathbb R^m \to \mathbb R^k$ such that
	\begin{equation*}
		\big(\mu(S_1), \hdots, \mu(S_k)\big) = f\big(\mu(\tilde S_1), \hdots, \mu(\tilde S_m) \big)
	\end{equation*}
	for all $\mu \in \mathbf M(K)$. Hence, the statement follows by the continuous mapping theorem.
\end{proof}
In order to apply \cref{Lemma: Fidi convegene} to our case, we will need the following generalization of the Poisson limit theorem (which is a special case of \cite[Proposition 1.4]{LastPenrose.2017}).
\begin{lem}
	\label{Lemma: Generalized Poisson limit theorem}
	Assume that $X_{n, 1}, \hdots, X_{n, m_n}$ are for any $n\in \mathbb N$ independent $\{0, 1\}$-valued random variables such that
	\begin{equation*}
		\lim_{n \to \infty} \max_{1\leq i \leq m_n} \mathbb P(X_{n, i} = 1) = 0 \quad \text{ and }\quad \lim_{n \to \infty} \sum_{i=1}^{m_n} \mathbb P(X_{n, i} = 1) = \gamma > 0.
	\end{equation*}
	Then for any $k\in \mathbb N_0$
	\begin{equation*}
		\lim_{n \to \infty} \mathbb P\Big( \sum_{i=1}^{m_n} X_{n, i} = k\Big) = \operatorname{Poi}(\gamma) (\{k\}) \coloneqq \frac{\ee^{-\gamma} \gamma^k}{k!}.
	\end{equation*}	
\end{lem}
We can now prove the convergence of the point processes defined in \cref{sec:longrangeorder} and the beginning of this section, respectively.
\begin{prop}
	\label{Propostion: Convergence in distribution of point processes}
	We have $(\xi_{N, 1}, \xi_{N, 2}) \to (\xi_1, \xi_2)$ in distribution as $N\to \infty$.
\end{prop}
\begin{proof}
	We set $\mathcal G_1 \coloneqq \mathcal G([0, T])$ and $\mathcal G_2 \coloneqq \mathcal G([0, T]^2)$.
	By independence of $\xi_{N, 1}$ and $\xi_{N, 2}$, it is sufficient to show that $\xi_{N, 1} \to \xi_1$ and $\xi_{N, 2} \to \xi_2$ in distribution as $N\to \infty$.
	By \cref{Lemma: Fidi convegene}, it is further sufficient to show for all $i  \in \{1, 2\}$, $k\in \N$ and pairwise disjoint $I_1, \hdots, I_k \in \mathcal G_i$
	\begin{equation*}
		(\xi_{N, i}(I_1), \hdots, \xi_{N, i}(I_k)) \to (\xi_{i}(I_1), \hdots, \xi_{i}(I_k))
		\quad\text{in distribution.}
	\end{equation*}
	By independence of the coordinates, it is even sufficient to show that $\xi_{N, i}(I) \to \xi_{i}(I)$ in distribution for all  $i\in \{1, 2\}$ and $I\in \mathcal G_i$.
	
	For the case $i=1$, we fix $0\leq a<b \leq T$ and notice that
	\begin{equation*}
		\xi_{N, 1}\big((a, b]\big) \sim \operatorname{Bin}\big(n_N, 2\delta(T, N)\big) \text{ with }  n_N \coloneqq \big|(a, b] \cap \Lambda(T, N)\big|.  
	\end{equation*}
	Now $2\delta(T, N) \cdot n_N \xrightarrow{N\to\infty} 2(b-a)$ and hence, by the (standard) Poisson limit theorem,
	\begin{equation*}
		\xi_{N, 1}\big((a, b]\big) \xrightarrow{N\to\infty} \xi_{1}\big((a, b]\big) \sim \operatorname{Poi}\big(2(b-a)\big) \quad \text{in distribution.}
	\end{equation*}
	For the case $i=2$, let us fix $I\in \mathcal G_2$.
	Using the definition \cref{def:paN}, the estimate \cref{Equation: Estimate exponential}  and
	that $g$ is bounded,
	we obtain
	\begin{equation*}
		\sum_{i, j \in \Lambda(T, N)\cap I} \chr_{\{i \nsim j, \,i<j \}}  p_{\alpha, N}(i, j) =
		\Big( 2 \alpha \sum_{ i, j\in \Lambda(T, N) \cap I} \chr_{\{i<j \}} \delta(T, N)^2 g(j-i) \Big) + \mathcal O(1/N)
	\end{equation*} 
	leading to 
	\begin{equation*}
		\sum_{i, j \in \Lambda(T, N)\cap I} \chr_{\{i \nsim j, \,i<j \}} p_{\alpha, N}(i, j) \xrightarrow{N\to\infty} 2 \alpha \int_I g(t-s) \chr_{\{0\leq s<t \leq T\}}\d s \d t.
	\end{equation*}
	Additionally,
	\begin{equation*}
		\max_{i, j \in \Lambda(T, N)} p_{\alpha, N}(i, j) \leq 2\delta(T, N) \vee \Big(1- \ee^{-2 \alpha \delta(T, N)^2 g(0)} \Big) \xrightarrow{N\to\infty} 0.
	\end{equation*}
	Hence, with \cref{Lemma: Generalized Poisson limit theorem}, we obtain
	\begin{equation*}
		\xi_{N, 2}(I) \xrightarrow{N\to\infty}  \xi_{2}(I) \sim \operatorname{Poi}\Big(2 \alpha \int_I g(t-s)  \chr_{\{0\leq s<t \leq T\}} \, \d s \d t \Big)
		\quad
		\text{in distribution.}\qedhere
	\end{equation*}
\end{proof}
To translate the above proposition into the convergence of the respective two-point functions, we will use the following simple rearrangement lemma for pure point measures which are close in the Prohorov distance.
\begin{lem}
	\label{Lemma: Distance in Prohorov metric}
	If the Prohorov-distance between $\mu = \sum_{i=1}^n \delta_{x_i} \in \mathbf N(K)$ and $\nu = \sum_{j=1}^m \delta_{y_j} \in \mathbf N(K)$ satisfies
	\begin{equation*}
		d(\mu, \nu)  < \varepsilon < 1 \wedge \min \{\tfrac{1}{2}|x-y|:\, x, y \in \operatorname{supp}(\mu),\, x\neq y\}
	\end{equation*}
	then $n=m$ and there exists a permutation $\sigma: \{1, \hdots, n\} \to \{1, \hdots, n\}$ such that
	\begin{equation*}
		\max_{1\leq i \leq n} |x_i - y_{\sigma(i)}| <\varepsilon.
	\end{equation*}
\end{lem}
\begin{proof}
	Since $|\mu(K) - \nu(K)|< 1$ and $\mu(K), \nu(K) \in \mathbb N_0$, we immediately obtain $\mu(K) = \nu(K)$, i.e., $n=m$.
	Now let $x\in \operatorname{supp}(\mu)$.
	Then, using the notation defined below \cref{def:prohorov}, we have
	\begin{align*}
		\nu(\{x\}_\eps) \ge \mu(\{x\})-\eps, \quad \text{and thus}\quad \nu(\{x\}_\eps)\ge\mu(\{x\}),
	\end{align*}
	since $\nu(\{x\}_\eps), \mu(\{x\}) \in \mathbb N_0$.
	From $n=m$ and $\{x\}_\eps \cap \{y\}_\eps = \emptyset$ for $x, y \in \operatorname{supp}(\mu)$ with $x\neq y$, we obtain $\nu(\{x\}_\eps) = \mu(\{x\})$ for all $x\in \operatorname{supp}(\mu)$ and thus the statement.
\end{proof}
We can now prove \cref{eq:percolim}.
\begin{prop}
	\label{Proposition: Convergence to continuum percolation}
	For all $T\in \N$ and $n\in [0, T] \cap \N_0$, we have $\lim_{N\to \infty}\mathbf P_{\alpha, T, NT}(0 \leftrightarrow n) = \mathbf P_{\alpha, T,\sfc}(0 \leftrightarrow n)$. 
\end{prop}
\begin{proof}
	By \cref{Propostion: Convergence in distribution of point processes} and the Portmanteau Theorem, it is sufficient to show that $\{0 \leftrightarrow n\}$ is a continuity set of the distribution $\mathbf P_{\alpha, T,\sfc}$ of $\xi$.
	
	We call a configuration $(\mu_1, \mu_2) \in \mathbf N([0, T]) \times \mathbf N([0, T]^2)$ stable if $\mu_1$ only has atoms of mass 1 and if $x\cap\{n,s,t\}=\emptyset$ holds for all $x\in \operatorname{supp}(\mu_1)$ and $(s, t) \in \operatorname{supp}(\mu_2)$.
	If $(\mu_1, \mu_2)$ is stable, then \cref{Lemma: Distance in Prohorov metric} yields that $(\mu_1, \mu_2)$ is an interior point of either $\{0 \leftrightarrow n\}$ or $\{0 \nleftrightarrow n\}$ (w.r.t. to the product topology).
	Hence,
	\begin{equation*}
		\mathbb P\big(\xi \in \partial \{0 \leftrightarrow n\} \big) \leq \mathbb P(\xi \text{ is not stable}) = 0. \qedhere
	\end{equation*}
\end{proof}

\begin{rem}
	By extending the arguments above, one can represent spin correlations under $\mathbb P_{\alpha, T}$ in terms of a continuum FK-percolation by taking the limit $N\to \infty$ in \cref{Equation: Correlations in terms of FK-Percolation}. Let $\tilde \xi_1, \tilde \xi_2$ be independent Poisson point processes with intensity measures $\chr_{\{0\leq x\leq T\}}\d x$ and $4\alpha g(t-s) \chr_{\{0\leq s<t\leq T\}}\, \d s \d t$ respectively and let $\widetilde{\bfP}_{\alpha, T, \sfc}$ be the law of $\tilde \xi = (\tilde \xi_1, \tilde \xi_2)$. Let $\tilde \xi_N = (\tilde \xi_{N, 1}, \tilde \xi_{N, 2})$ be its discrete version, defined in the same manner as $(\xi_{N, 1}, \xi_{N, 2})$ after replacing $p_{\alpha, N}$ by $\tilde p_{\alpha, N}$ (see \cref{def:tildepaN,def:paN} for the definitions). Then $\tilde \xi_N \xrightarrow{N\to\infty} \tilde \xi$ in distribution.
	We denote by $\cP_{\alpha, T, \sfc}$ the law of the continuum FK-percolation on $\bfN([0, T]) \times \bfN([0, T]^2)$ defined by 
	\begin{equation*}
		\cP_{\alpha, T, \sfc}(\d\xi) \propto 2^{\cC (\xi)} \widetilde{\bfP}_{\alpha, T, \sfc}(\d \xi)
	\end{equation*}
	where $\cC$ denotes the number of connected components.
	The set of discontinuity points of $2^\cC$ is a null set under $\widetilde{\bfP}_{\alpha, T, \sfc}$. Using that
	\begin{equation*}
		\cC(\tilde \xi_N)  \leq \tilde \xi_{N, 1}([0, T]) \sim \operatorname{Bin}\big(N, \delta(T, N) \big)
	\end{equation*}
	one can show that the sequence $(2^{\mathcal C(\tilde \xi_N)})_N$ is uniformly integrable. With the Portmanteau Theorem, one hence obtains $\mathcal P_{\alpha, T, N} \to \mathcal P_{\alpha, T, \sfc}$ weakly, and, by taking the limit $N\to \infty$ in \cref{Equation: Correlations in terms of FK-Percolation},
	\begin{equation*}
		\mathbb E_{\alpha, T}[X_0 X_t] = \mathcal P_{\alpha, T, \sfc}(0 \leftrightarrow t). 
	\end{equation*}
\end{rem}


\bibliographystyle{halpha-abbrv}
\bibliography{../../Literature/00lit}

\end{document}